\documentclass[12pt]{article}
\usepackage[top=24mm, bottom=24mm, left=24mm, right=24mm]{geometry}
\usepackage{amsfonts,amsmath,amsthm,amssymb,bm,bbm,dsfont}
\usepackage{lmodern}
\usepackage{enumitem,placeins}
\usepackage{natbib,color}
\usepackage[pdftex]{graphicx}
\usepackage{subfigure,float,booktabs,longtable}
\usepackage{dcolumn}
\usepackage{changepage}
\usepackage{mathtools}
\usepackage{setspace}
\newcolumntype{L}{D{.}{.}{2,3}}
\allowdisplaybreaks

\definecolor{winered}{rgb}{0.5,0,0}
\usepackage[pdfstartview=FitH,  bookmarks=true,  bookmarksopen=true,  breaklinks=true,  colorlinks=true,  linkcolor=winered,anchorcolor=my,  citecolor=blue, filecolor=blue,  menucolor=blue,  urlcolor=black]{hyperref}

\numberwithin{equation}{section}
\newtheorem{theorem}{Theorem}[section]
\newtheorem{lemma}[theorem]{Lemma}
\newtheorem{proposition}[theorem]{Proposition}
\newtheorem{corollary}[theorem]{Corollary}

\theoremstyle{definition}

\newtheorem{example}{Example}
{\theoremstyle{plain}
	\newtheorem{assumption}{Assumption}}
\definecolor{my}{rgb}{0.05,0.05,0.5}
\definecolor{myBlue}{rgb}{.1,.1,.5}
\definecolor{myGreen}{rgb}{0,.4,0}
\definecolor{myRed}{rgb}{.25,0.15,.5}

\definecolor{my}{rgb}{0.05,0.05,0.5}

\newcommand{\conp}{\stackrel{p}{\longrightarrow}}

\renewcommand{\limsup}{\displaystyle \operatornamewithlimits{\lim\sup \ }}
\renewcommand{\liminf}{\displaystyle \operatornamewithlimits{\lim\inf \ }}
\renewcommand{\mathbf}[1]{\textbf{\textit{#1}}}

\newcommand{\E}{\operatorname{E}}

\newcommand{\Rmnum}[1]{\expandafter\@slowromancap\romannumeral #1@}

\newcommand{\sm}{\mathtt{s}_M}
\newcommand{\hsm}{\hat{\mathtt{s}}_M}

\newfont{\bbf}{cmbx12 scaled 1435}

\begin{document}
	
\title{{Selective Inference for Survey-Weighted Generalized Linear Models: Growing Dimension and Average Marginal Effects}\thanks{The authors thank Denis Chetverikov, the co-editor, and two anonymous referees for their constructive comments and suggestions, which substantially improved the paper; Lucas Van Den Driesschen for research assistance; and participants at the 2024 Canadian Economics Association (CEA) Conference for helpful discussions. The authors 
acknowledge financial support of the Catalyzing Interdisciplinary Research Clusters 
Initiative (CIRC), York University. This research is part of the projects ``Data Economy'' and ``Digital Currencies,'' approved and funded by the office of VPRI and participating faculties of York University. The authors also acknowledge access to the data provided by RDC and 
Statistics Canada (Project 21-MAPA YRK-721).}}
\date{\today}
\author{
Joann Jasiak\thanks{York University, \texttt{jasiakj@yorku.ca}}
\and 
Purevdorj Tuvaandorj\thanks{York University, \texttt{tpujee@yorku.ca}}
}
\maketitle
\begin{abstract}
This paper extends the $\ell_1$-penalized likelihood Selective Inference method to survey settings. We establish its asymptotic validity for generalized linear models with an increasing number of covariates and survey weights, allowing inference on pseudo-true parameters of possibly misspecified selected models. Moreover, we generalize the method to facilitate inference on nonlinear functionals, such as the average marginal effect in survey-weighted logit models. The effectiveness of the proposed approach is illustrated using simulated data and the 2020 Canadian Internet Use Survey.
	\begin{description}
		\item[Keywords:]Average Marginal Effect, Growing dimension, $\ell_1$-regularization, Logit, Selective Inference, Survey data, Yurinskii coupling 
	\end{description}
\end{abstract}
\section{Introduction}

Survey data are widely used across many disciplines in the social sciences. The statistical methodology for analyzing survey samples has been extensively developed, resulting in a large body of literature \citep[see, e.g.,][]{Cameron-Trivedi(2009), Wooldridge(2010), Fuller(2011), Thompson(2012)}.

Despite the rapid development of machine learning and high-dimensional econometric methods, together with the increasing availability of large datasets, relatively little work has addressed how these methods can be adapted to survey data. This paper contributes to filling this gap by extending the $\ell_1$-penalized likelihood Selective Inference (SI) method to a survey setting. Our goal is to expand the toolbox available to applied researchers who wish to employ high-dimensional regression techniques with survey data.

Penalized prediction-oriented methods, including the Lasso, involve a bias-variance trade-off and produce estimates that are not directly suitable for conventional inference on model coefficients. Several methods addressing the inferential difficulties associated with regularization and model selection have been proposed in the literature. Among these, \cite{Zhang-Zhang(2014)} and \cite{Javanmard-Montanari(2014)} propose debiased Lasso procedures based on a one-step correction of the initial Lasso estimator. \cite{belloni2014inference} and \cite{Belloni-Chernozhukov-Wei(2016)} develop double-selection and $C(\alpha)$-type procedures for generalized linear models and partially linear models under sparsity assumptions. A related recent contribution is \cite{ChetverikovSorensen2025}, who propose a tuning parameter selection method termed bootstrapping after cross-validation (BCV). Another post-selection inference approach is the PoSI method proposed by \cite{Berk2013}. This approach formulates post-selection inference as a simultaneous inference problem over a prespecified family of models.

An inferential procedure conceptually distinct from these approaches is SI, which is the focus of this paper. The SI framework for linear models selected by the Lasso was developed by \cite{Lee-Sun-Sun-Taylor(2016)} and subsequently extended by \cite{Taylor-Tibshirani(2018)} to generalized linear models (GLMs) estimated by $\ell_1$-penalized maximum likelihood (PML). In SI, inference on the parameters of the selected model is conducted conditional on the selection event. This distinguishes SI from PoSI, which constructs simultaneous confidence intervals across a prespecified family of models without conditioning on the particular model selected.

In applications of this kind, the $\ell_1$-PML serves as a model-selection device rather than as an end in itself. Once the data have been used to choose a specification, the questions of interest are which selected associations remain distinguishable from zero after accounting for selection, how large they are, and how precisely they are estimated. Conventional Wald intervals computed from the selected model ignore the selection step and may not attain their nominal coverage. SI answers these questions for the population parameters of the selected model, which we define as pseudo-true values so that they remain well defined when the selected model omits relevant regressors. These targets differ from those of the debiased Lasso, double-selection, and other Neyman-orthogonal procedures such as the $C(\alpha)$ test, which deliver inference on a low-dimensional parameter of the full model specified before selection. As with any regression-based summary, the SI targets are associational; a causal interpretation requires additional identifying assumptions.

This paper provides two main results. First, we extend the SI framework to survey-weighted GLMs estimated via $\ell_1$-PML and establish its asymptotic validity when the number of covariates grows with the sample size.\footnote{We thank the co-editor for suggesting this extension using Yurinskii's coupling.} Our survey framework follows \cite{Wooldridge(2001)}. With survey weights, the usual information matrix equality fails to hold, and the statistics associated with the active and inactive parts of the Karush--Kuhn--Tucker (KKT) conditions need not be asymptotically independent. Consequently, the simplification used by \cite{Taylor-Tibshirani(2018)}, which relies on this independence, is not generally applicable. We therefore conduct inference conditional on the full selection event, characterized jointly by the active and inactive KKT conditions. Second, we establish the asymptotic validity of the proposed procedure for inference on nonlinear functionals, including average marginal effects (AMEs) in survey-weighted logit models.

Several recent papers study penalized estimation and bias-corrected inference for related quantities. \cite{WooldridgeZhu2020} develop inference on average partial effects (APEs) in approximately sparse correlated random effects probit models with panel data; \cite{HirshbergWager2020} propose debiased inference on APEs in single-index models; and \cite{Chiang2022} studies inference on many APEs with an application to text regression. Unlike these approaches, we develop SI for AMEs in survey-weighted GLMs.

There are relatively few studies applying Lasso-type methods in survey settings. For example, \cite{Mcconville-etal(2017)} study model-assisted survey regression estimation using the Lasso and establish finite-population asymptotic results with a fixed number of regressors. In contrast, we consider survey-weighted GLMs and establish the asymptotic validity of the SI procedure for pseudo-true parameters of selected models, without requiring sure screening, under the standard infinite-population framework commonly used in econometrics (see, e.g., \cite{Wooldridge(2001), Wooldridge(2010), Cameron-Trivedi(2009)}).

The remainder of the paper is organized as follows. Section \ref{sec: model} introduces the model framework. Section \ref{sec: infer} develops the selective inference procedure for survey-weighted GLMs. Section \ref{sec: simul} presents simulation evidence on the properties of the proposed method, and Section \ref{sec: app} provides an empirical application using the 2020 Canadian Internet Use Survey. Section \ref{sec: conc} concludes.

\paragraph*{Notations and terminology}
Let $1(\cdot)$ denote the indicator function, and $\lambda_{\min}(A)$ and $\lambda_{\max}(A)$ denote the smallest and largest eigenvalues of a symmetric matrix $A$, respectively. For a $k\times1$ vector $a=(a_1,\dots,a_k)'$, define
$
\operatorname{supp}(a)\equiv\{j\in\{1,\dots,k\}:a_j\neq0\},$
$
\|a\|_0\equiv|\operatorname{supp}(a)|,$
$
\|a\|_1\equiv\sum_{j=1}^k|a_j|.$ For a real matrix 
$A=(a_{ij})$, let $\Vert A\Vert_{\infty}\equiv \max_{i,j}\left|a_{ij}\right|$ denote its entrywise maximum norm, and let $\Vert A\Vert=\sqrt{\mathrm{tr}(A'A)}$ 
and $\Vert A\Vert_2=\sqrt{\lambda_{\max}(A'A)}$ denote its Frobenius and spectral norms, respectively. 
The sub-Gaussian norm of a random variable $X$ is defined as 
$\Vert X\Vert_{\psi_2}\equiv \sup_{m\geq 1}m^{-1/2}(\E[\vert X\vert^m])^{1/m}.$
A random variable $X$ is called sub-Gaussian if $\Vert X\Vert_{\psi_2}\leq C<\infty$ for a constant 
$C>0$. 	A random vector $X\in \mathbb{R}^p$ is called sub-Gaussian if the one-dimensional marginals $X'b$ are sub-Gaussian random variables for all $b\in\mathbb{R}^p$. The sub-Gaussian norm for the random vector is defined as $\Vert X\Vert_{\psi_2}\equiv \sup_{\Vert b\Vert=1}\Vert X'b\Vert_{\psi_2}$. 
The sub-exponential norm of a random variable $X\in \mathbb{R}$ is 	$\Vert X\Vert_{\psi_1} \equiv \inf \{t>0: \E[\exp(|X|/t)]\leq 2\}$. Moreover, define
$\bm{1}_{m}\equiv (1,\dots, 1)'\in\mathbb{R}^m$, $\bm{0}_{m}\equiv (0,\dots, 0)'\in\mathbb{R}^m$, and let $\bm{0}_{l\times m}$ denote the $l\times m$ matrix of zeros, and $\bm{e}_k$ denote the unit vector whose $k$-th element is $1$ and all remaining elements are $0$ (with the dimension understood from the context). We use $I_r$ to denote the $r\times r$ identity matrix.
Let $F(x; \mu, \sigma^2, a,b)$ denote the CDF of an $N(\mu,\sigma^2)$ random variable truncated on the interval $[a,b]$: for $x\in[a,b]$,
$
F(x; \mu, \sigma^2, a,b)\equiv {\left\{\Phi((x-\mu)/\sigma)-\Phi((a-\mu)/\sigma)\right\}}
/{\left\{\Phi((b-\mu)/\sigma)-\Phi((a-\mu)/\sigma)\right\}},
$
where $\Phi(\cdot)$ denotes the standard normal CDF. The CDF is $0$ for $x<a$ and $1$ for $x>b$. Also, let $\Lambda(z)\equiv \exp(z)/(1+\exp(z))$ denote the CDF of the standard logistic distribution. For scalar functions of a vector argument,
$\nabla$ denotes the column-vector gradient. We abbreviate the central limit theorem, the continuous mapping theorem, and ``with probability approaching one'' as CLT, CMT, and wpa 1, respectively.

\section{Model}\label{sec: model}

We consider a GLM that specifies the conditional density of a scalar 
outcome variable $y_i$ given a $(d+p)\times 1$ vector of covariates 
$x_i$ as
\begin{equation*}
f(y_i\vert x_i,\theta_0)
=\exp\,\left(y_ix_i'\theta_0-a(x_i'\theta_0)\right)c(y_i),
\quad i=1,\dots,n,
\end{equation*}
where $\theta_0$ is the true value of the parameter vector 
$\theta\in\mathbb{R}^{d+p}$, and $a(\cdot)$ and $c(\cdot)$ are known 
functions. To each observation $(y_i,x_i')'$ there corresponds a 
positive, bounded, and possibly random survey weight $w_i$, which may depend
on the unit's own characteristics, for example through its stratum.\footnote{In our framework, 
$\{(y_i, x_i', w_i)'\}_{i=1}^n$ forms a triangular array $\{(y_{ni}, x_{ni}', w_{ni})'\}_{i=1}^n$. We drop the index $n$ for notational simplicity.} Let $g(y,x'\theta)\equiv -\log f(y,x'\theta)$ and define the weighted 
log-likelihood function as
\begin{equation}\label{eq: WLL}
L(\theta)\equiv -n^{-1}\sum_{i=1}^n w_i\,g(y_i,x_i'\theta).
\end{equation}
The weighted likelihood framework is widely used in survey data analysis 
\citep{Manski-Lerman(1977),Cameron-Trivedi(2009),Wooldridge(2010)} and 
accommodates, among others, the following stratification schemes.
\begin{example}[Standard stratified sampling]\label{ex: SS}
Let $\mathcal{Z}$ be the population for $z=(y,x')'$, assumed infinite 
(or containing a large number of units). $\mathcal{Z}$ is stratified 
into $J$ nonempty, mutually exclusive and exhaustive strata 
$\mathcal{Z}=\bigcup_{j=1}^J\mathcal{Z}_j$. A sample of $n_j$ 
observations $\{(y_{ij},x_{ij}')'\}_{i=1}^{n_j}$ is drawn independently 
from stratum $\mathcal{Z}_j$, $j=1,\dots,J$. The strata sample sizes 
$n_j$ are non-random, and the population frequencies 
$q_j=P[z\in\mathcal{Z}_j]>0$ are known. Weights for stratum $j$ are 
$w_i=q_j/(n_j/n)$ for $i\in\text{stratum }j$, with $n=\sum_{j=1}^J n_j$.
Pooling and reindexing the observations as $i=1,\ldots,n$, the corresponding weighted log-likelihood function is
\begin{equation}\label{eq: SS loglik}
L(\theta)
=-\sum_{j=1}^J q_j\!\left(n_j^{-1}\sum_{i=1}^{n_j}
g(y_{ij},x_{ij}'\theta)\right)
=-n^{-1}\sum_{i=1}^n w_i\,g(y_i,x_i'\theta).
\end{equation}
\end{example}

\begin{example}[Exogenous stratification]\label{ex: ES}
Let $\mathcal{Z}=\mathcal{Y}\times\mathcal{X}$. The population is 
stratified into $J$ strata according to a deterministic function of 
$x_i$: $\mathcal{X}=\cup_{j=1}^J\mathcal{X}_j$, where 
$\mathcal{X}_j$, $j=1,\dots,J$, are mutually exclusive. Population 
frequencies $q_j=P[x\in\mathcal{X}_j]>0$ are known. Given 
$n=\sum_{j=1}^J n_j$ observations drawn independently from each stratum, 
the likelihood function takes the form \eqref{eq: SS loglik}.
\end{example}
\noindent\cite{Wooldridge(2001)} established the asymptotic properties of $M$-estimators under these sampling schemes. We establish the asymptotic validity of the proposed SI procedure more generally for independent, non-identically distributed triangular arrays, without requiring either particular sampling scheme.\par
The score function and negative Hessian corresponding to \eqref{eq: WLL} are
\begin{align}
 S(\theta)
 &\equiv\frac{\partial L(\theta)}{\partial\theta}
 =-n^{-1}\sum_{i=1}^n w_i x_i\dot{g}(y_i,x_i'\theta),
 \quad\dot{g}(y,t)\equiv\frac{\partial g(y,t)}{\partial t},
 \label{eq: scoref}\\
 \hat{H}(\theta)
 &\equiv-\frac{\partial^2 L(\theta)}{\partial\theta\,\partial\theta'}
 =n^{-1}\sum_{i=1}^n w_i x_i x_i'\ddot{g}(y_i,x_i'\theta),
 \quad\ddot{g}(y,t)\equiv\frac{\partial^2 g(y,t)}{\partial t^2}.
 \label{eq: hessian}
\end{align}
\paragraph*{Survey weights and information equality.}
Since $g(y,t)=-yt+a(t)-\log c(y)$, we have $\dot g(y,t)=a'(t)-y$ and
$\ddot g(y,t)=a''(t)$. Suppose, for illustration, that the model is
correctly specified and that $w_i$ is a function of $x_i$. Then the score
contribution $w_ix_i\dot g(y_i,x_i'\theta_0)$ has conditional mean zero and
conditional variance $w_i^2a''(x_i'\theta_0)x_ix_i'$, whereas the
corresponding contribution to the negative Hessian is
$w_ia''(x_i'\theta_0)x_ix_i'$. Each conditional variance thus equals
$w_i$ times the corresponding Hessian contribution, and this factor varies
across observations. Hence
$n^{-1}\sum_iw_i^2a''(x_i'\theta_0)x_ix_i'$ is in general not equal, or
even proportional, to $n^{-1}\sum_iw_ia''(x_i'\theta_0)x_ix_i'$. We refer
to this feature as heteroskedasticity of the weighted score. As a
consequence, the
information matrix equality
$\operatorname{Var}(n^{1/2}S(\theta_0))=\E[\hat H(\theta_0)]$ fails in
general, even under correct specification. Section~\ref{subsub: selective}
shows that this failure makes the random quantities that define the
active and inactive KKT constraints statistically dependent.
\paragraph*{Parameter partition.}
Partition the covariate and parameter vectors as
$x_i=(x_{1i}',x_{2i}')'\in\mathbb{R}^{d+p}$ and $\theta=(\alpha',\beta')'\in\mathbb{R}^{d+p},$
where $x_{1i}\in\mathbb{R}^d$ collects the constant and any 
pre-specified regressors of interest, $\alpha\in\mathbb{R}^d$ are 
the corresponding unpenalized parameters, $x_{2i}\in\mathbb{R}^p$ 
are candidate regressors subject to selection, and 
$\beta\in\mathbb{R}^p$ are their coefficients. The linear index is 
$x_i'\theta=x_{1i}'\alpha+x_{2i}'\beta$. The special case $d=1$, 
$x_{1i}=1$ recovers the intercept-only formulation of the previous 
literature.
\paragraph*{$\ell_1$-PML.}
Variable selection and initial estimation are performed using the
survey-weighted $\ell_1$-PML estimator, obtained by minimizing the
penalized negative log-likelihood
\begin{equation}\label{eq: logitlasso}
	\min_{\theta=(\alpha',\beta')'\in\mathbb{R}^{d+p}}
	\Bigl(-L(\theta)+\lambda\Vert\beta\Vert_1\Bigr),
\end{equation}
where $\lambda\geq 0$ is a tuning parameter. Only the coefficients $\beta=(\beta_1,\dots,\beta_p)'$ are penalized; $\alpha$ is always estimated without regularization. Henceforth, the $j$-th 
elements of $\theta$ and $\theta_0$ are denoted as $\theta_{(j)}$ and 
$\theta_{0(j)}$, respectively.
\paragraph*{Selected model.}
For a given $\lambda$, the solution to \eqref{eq: logitlasso} 
produces a random index set
\begin{equation*}
	\hat{M}\equiv\{1,\dots,d\}
	\cup\{d+j: \hat{\beta}_j\neq 0,\; j=1,\dots,p\}
	\subseteq\{1,\dots,d+p\},
\end{equation*}
which we call the \emph{selected model}. The first $d$ indices, 
corresponding to the unpenalized $\alpha$, are always included in 
$\hat{M}$; the remaining indices in $\hat{M}$ are the 
components of $\beta$ with non-zero $\ell_1$-PML estimates.\par

For a generic (non-random) index set
$M\subseteq\{1,\dots,d+p\}$ with
$\{1,\dots,d\}\subseteq M$, let $|M|$ denote its cardinality
and $x_{iM}\in\mathbb{R}^{|M|}$ denote the subvector of $x_i$
corresponding to the indices in $M$. Let $X=[x_1,\ldots,x_n]'$ denote
the $n\times(d+p)$ design matrix and
$X_M=[x_{1M},\ldots,x_{nM}]'$ its submatrix corresponding to $M$.
Write
$\theta_M=(\alpha',\beta_M')'\in\mathbb{R}^{|M|}$
and
$\hat{\theta}_M=(\hat{\alpha}',\hat{\beta}_M')'
\in\mathbb{R}^{|M|}$, where $\beta_M\in\mathbb{R}^{|M|-d}$ and 
$\hat{\beta}_M\in\mathbb{R}^{|M|-d}$ are the subvectors of $\beta$ 
and $\hat{\beta}$ corresponding to the selected indices in 
$M\setminus\{1,\dots,d\}$. Define $\hsm\equiv\mathrm{sign}(\hat{\beta}_M)$.
On the event $\{\hat M=M\}$,
$\hsm\in\{-1,1\}^{|M|-d}$.
\paragraph*{Target parameter.}
For a generic model $M$, let $\theta_{0,M}$ denote the coordinate
subvector of the full-model parameter $\theta_0$ corresponding to $M$.
The target of selective inference (see Section~\ref{subsub: selective} below) is
instead the population parameter of the restricted model,
\begin{equation}\label{def: proj}
	\theta_M^*
	=\arg\min_{\theta_M\in\mathbb{R}^{|M|}}
	\mathcal{E}_M(\theta_M),
	\qquad
	\mathcal{E}_M(\theta_M)
	\equiv\E\!\left[n^{-1}\sum_{i=1}^n
	w_i\,g(y_i,x_{iM}'\theta_M)\right].
\end{equation}
Here and below, expectations are taken with respect to the joint
distribution of the observations $(y_i,x_i',w_i)'$, $i=1,\ldots,n$,
including the survey weights when they are random.

Write
$\theta_M^*=(\alpha_M^{*\prime},\beta_M^{*\prime})'$.
In general, $\theta_M^*$ need not equal $\theta_{0,M}$,
since omitting an active regressor can change the population
coefficients on the included regressors.
Let $\mathcal S_0\equiv\operatorname{supp}(\theta_0)$.
If $\mathcal S_0\subseteq M$ and $\E[S(\theta_0)]=0$,
then $\theta_{0,M}$ minimizes $\mathcal E_M$.
When this minimizer is unique, it follows that
\begin{equation}\label{eq: target equals subvector}
	\theta_M^*=\theta_{0,M}.
\end{equation}
The proposed SI procedure targets $\theta_M^*$ and does not
require $\mathcal S_0\subseteq M$. In particular, always
including the first $d$ coordinates does not by itself imply
$\alpha_M^*=\alpha_0$; that equality follows under the
conditions stated above. Let
$-M\equiv\{1,\ldots,d+p\}\setminus M.$
Because $\{1,\ldots,d\}\subseteq M$, the indices in $-M$ correspond
only to penalized regressors. For a restricted parameter
$\theta_M\in\mathbb R^{|M|}$, define $L_M(\theta_M)
\equiv-n^{-1}\sum_{i=1}^n
w_i g(y_i,x_{iM}'\theta_M),$
\begin{align*}
S(\theta_M)
&\equiv-n^{-1}\sum_{i=1}^n
w_i x_i\dot g(y_i,x_{iM}'\theta_M),\quad 
\hat H(\theta_M)
\equiv n^{-1}\sum_{i=1}^n
w_i x_ix_i'\ddot g(y_i,x_{iM}'\theta_M).
\end{align*}
Thus, $S(\theta_M)$ and $\hat H(\theta_M)$ denote the
full-dimensional score and negative Hessian evaluated
at the full parameter vector with coordinates outside
$M$ set to zero. For $A,B\in\{M,-M\}$, define
$S_A(\theta_M)\equiv[S(\theta_M)]_A$ and
$\hat H_{A,B}(\theta_M)\equiv[\hat H(\theta_M)]_{A,B}$.
Accordingly, after ordering the coordinates as $(M,-M)$,
\begin{equation*}
S(\theta_M)=
\begin{bmatrix}
S_M(\theta_M)\\
S_{-M}(\theta_M)
\end{bmatrix},
\qquad
\hat H(\theta_M)=
\begin{bmatrix}
\hat H_{M,M}(\theta_M) & \hat H_{M,-M}(\theta_M)\\
\hat H_{-M,M}(\theta_M) & \hat H_{-M,-M}(\theta_M)
\end{bmatrix}.
\end{equation*}
We abbreviate
$\hat H_M(\theta_M)\equiv\hat H_{M,M}(\theta_M)$.
Here $S_M(\theta_M)\in\mathbb R^{|M|}$ is the score for
the restricted model, while $S_{-M}(\theta_M)$ records
the score in the excluded directions. By the definition of $\theta_M^*$, its
population first-order condition is
\begin{equation}\label{eq: pseudo score}
	\E[S_M(\theta_M^*)]=0.
\end{equation}
Throughout, we assume that the optimization problem
\eqref{eq: logitlasso} admits an almost surely unique solution.
Sufficient conditions for existence and uniqueness of
$\ell_1$-penalized likelihood estimators are provided by
\cite{AliTibshirani2019}. Consequently, $\hat\theta$,
$\hat M$, and the corresponding coefficient-sign vector
are well-defined almost surely.

For a fixed model $M$ containing $\{1,\dots,d\}$ and a fixed
sign vector $\sm\in\{-1,1\}^{|M|-d}$, the selection event
$\{\hat M=M,\hsm=\sm\}$ occurs if and only if there exist
$\hat{\theta}_M\in\mathbb{R}^{|M|}$ and
$\mathtt{u}\in\mathbb{R}^{d+p-|M|}$ satisfying the
Karush--Kuhn--Tucker (KKT) conditions for
\eqref{eq: logitlasso} (see \cite{Lee-Sun-Sun-Taylor(2016)}
and \cite{Taylor-Tibshirani(2018)}):
\begin{align}
S_M(\hat{\theta}_M)-(\bm{0}_d',\lambda\,\sm')'
&=0,
\quad \sm=\mathrm{sign}(\hat{\beta}_M),
\label{eq:KKT1}\\
S_{-M}(\hat{\theta}_M)-\lambda\,\mathtt{u}
&=0,
\quad\Vert\mathtt{u}\Vert_\infty\leq 1.
\label{eq:KKT2}
\end{align}

\paragraph*{Average marginal effects and scalar functionals.}
The conditional mean implied by the exponential family specification
is $\E[y_i|x_i]=a'(x_i'\theta_0)$.
For a continuous regressor $x_{ij}$ with coefficient $\theta_{0(j)}$,
the marginal effect on the conditional mean is
$
\frac{\partial}{\partial x_{ij}}\E[y_i|x_i]
=
a''(x_i'\theta_0)\,\theta_{0(j)}.
$
More generally, for a parameter vector $\theta$, define
$\mathrm{ME}_{ij}(\theta)=a''(x_i'\theta)\theta_{(j)}$
for a continuous regressor. For a binary regressor,
the marginal effect is defined as the discrete change
$
\mathrm{ME}_{ij}(\theta)
\equiv
a'(x_i'\theta)\big|_{x_{ij}=1}
-
a'(x_i'\theta)\big|_{x_{ij}=0}.
$
The corresponding AME is
\begin{equation*}
	\mathrm{AME}_j(\theta_0)
	\equiv
	\E\left[
	\frac{\sum_{i=1}^n w_i\,\mathrm{ME}_{ij}(\theta_0)}
	{\sum_{i=1}^n w_i}
	\right],
\end{equation*}
where the expectation is taken with respect to the joint
distribution of the regressors and survey weights.
A natural plug-in estimator replaces $\theta_0$ by an estimator
$\hat\theta=(\hat\alpha',\hat\beta')'$, for example the
survey-weighted $\ell_1$-PML estimator:
\begin{equation*}
	\widehat{\mathrm{AME}}_j(\hat\theta)
	=
	\frac{\sum_{i=1}^n w_i\,\mathrm{ME}_{ij}(\hat\theta)}
	{\sum_{i=1}^n w_i}.
\end{equation*}

For a selected model $M$, let
$\mathrm{ME}_{ij}(\theta_M)$ denote the marginal
effect of regressor $j\in M\setminus\{1,\ldots,d\}$,
computed as above under the restricted model
with linear index $x_{iM}'\theta_M$.
Let $\beta_{M,j}$ denote the corresponding
coefficient in $\theta_M$. For a continuous
regressor in the logit model,
$\mathrm{ME}_{ij}(\theta_M)
=\beta_{M,j}\Lambda'(x_{iM}'\theta_M)$.
The corresponding AME is
\begin{equation*}
	\mathrm{AME}_{Mj}(\theta_M)
	\equiv
	\E\left[
	\frac{\sum_{i=1}^n w_i\,
		\mathrm{ME}_{ij}(\theta_M)}
	{\sum_{i=1}^n w_i}
	\right].
\end{equation*}
We collect the AMEs into the vector
$\mathrm{AME}_M(\theta_M)=\left(\mathrm{AME}_{Mj}(\theta_M)\right)_{j\in M\setminus\{1,\ldots,d\}}\in\mathbb R^{|M|-d}.$

More generally, for a fixed model $M$, let
$\rho_M:\mathbb R^{|M|}\to\mathbb R$
denote a differentiable scalar population functional
of $\theta_M$, and let $\hat\rho_M$ denote its sample
counterpart. For the $k$-th selected slope coefficient,
take $\rho_M(\theta_M)=(\beta_M)_k$,
$k\in\{1,\ldots,|M|-d\}$; for an unpenalized
component, take $\rho_M(\theta_M)=\bm e_k'\alpha$,
$k\in\{1,\ldots,d\}$.
In both cases, $\hat\rho_M=\rho_M$.
For an AME, take
$\rho_M(\theta_M)=\mathrm{AME}_{Mj}(\theta_M)$,
with sample counterpart
\begin{equation*}
	\hat\rho_M(\theta_M)
	=
	\frac{\sum_{i=1}^n w_i\,
		\mathrm{ME}_{ij}(\theta_M)}
	{\sum_{i=1}^n w_i}.
\end{equation*}
The population target is $\rho_M(\theta_M^*)$,
and $\hat\rho_M(\hat\theta_M)$ is its plug-in estimator.

\section{Selective Inference}\label{sec: infer}

The objective of SI is to construct confidence intervals that
account for model selection. Since the selected model $\hat M$
is random, so is the population target
$\rho_{\hat M}(\theta_{\hat M}^*)$.
Moreover, $\theta_{\hat M}^*$ need not coincide with the
corresponding subvector $\theta_{0,\hat M}$ of the
full-model parameter. Conditional on $\{\hat M=M\}$,
however, the target is fixed at $\rho_M(\theta_M^*)$.

Conventional confidence intervals that ignore model selection
need not attain the desired conditional coverage, since selection
affects the sampling distribution of the estimator
\citep{Lee-Sun-Sun-Taylor(2016)}.
For a fixed model $M$ and scalar functional $\rho_M$,
we build a confidence interval $\mathrm{CI}_{\hat M}$ satisfying
\begin{equation}\label{eq: SI model objective}
	P\left[
	\rho_M(\theta_M^*)\in\mathrm{CI}_{\hat M}
	\middle|\hat M=M
	\right]\to1-\zeta.
\end{equation}
Following \cite{Lee-Sun-Sun-Taylor(2016)}, we first construct
confidence intervals by additionally conditioning on the signs
of the selected coefficients. This facilitates a polyhedral
representation of the selection constraints and allows us to establish
the sign-conditional coverage property
\begin{equation}\label{eq: SI objective}
	P\left[
	\rho_M(\theta_M^*)\in\mathrm{CI}_{\hat M}
	\middle|\hat M=M,\ \hsm=\sm
	\right]\to1-\zeta.
\end{equation}
Under suitable conditions, sign-conditional validity in
\eqref{eq: SI objective} then implies the model-conditional
coverage property in \eqref{eq: SI model objective}.

We establish the asymptotic validity of the proposed SI
procedure for survey-weighted GLMs with a growing number
of covariates, using the KKT characterization of the
selection event. For a deterministic sequence $m_n$
bounding the selected-model size, let
\begin{equation*}
\mathcal M_n(m_n)\equiv
\{M:\{1,\ldots,d\}\subseteq M,\ |M|\leq m_n\}.
\end{equation*}

\subsection{Construction of selective confidence intervals}
\label{subsub: selective}

\paragraph*{One-step estimator.}\phantomsection\label{subsec: rho SI}
For a
fixed model $M$, define the restricted penalized estimator by
\begin{equation}\label{eq: restricted lasso}
	\hat\theta_M\in
	\arg\min_{\theta_M=(\alpha',\beta_M')'}
	\left(-L_M(\theta_M)+\lambda\|\beta_M\|_1\right).
\end{equation}
On $\{\hat M=M,\hsm=\sm\}$, this estimator coincides with the active
block of the full $\ell_1$-PML estimator.
For inference on the scalar population functional
$\rho_M(\theta_M^*)$ introduced in Section~\ref{sec: model}, following \cite{Taylor-Tibshirani(2018)}, we define the one-step estimator of
$\rho_M(\theta_M^*)$ by
\begin{equation}\label{def: onestep rho}
	\tilde\rho_M\equiv\hat\rho_M(\hat\theta_M)
	+\nabla{\hat\rho}_M(\hat\theta_M)'
	\hat H_M(\hat\theta_M)^{-1}S_M(\hat\theta_M).
\end{equation}
For a parameter functional, this reduces to the usual
one-step correction. For an AME, it also retains the sampling
variation in the average of marginal effects. On the selection event, the KKT condition
\eqref{eq:KKT1} gives
$S_M(\hat\theta_M)=(\bm0_d',\lambda\sm')'$.
As in the classical one-step estimator
\citep[see, e.g.,][]{Newey-McFadden(1994)},
the correction removes the first-order penalization bias.\par 
Let $H_M\equiv\E[\hat H_M(\theta_M^*)]$.
Taylor expansions of the score and the sample counterpart
$\hat\rho_M$ around $\theta_M^*$ yield, to first order,
\begin{equation*}
	n^{1/2}\left(\tilde\rho_M-\rho_M(\theta_M^*)\right)
	\approx
	\nabla\rho_M(\theta_M^*)'H_M^{-1}
	n^{1/2}S_M(\theta_M^*)
	+n^{1/2}\left(\hat\rho_M(\theta_M^*)
	-\rho_M(\theta_M^*)\right).
\end{equation*}
The last term vanishes for selected coefficients and
unpenalized parameters, but is generally nonzero for
population AMEs.
To express the active and inactive selection constraints jointly,
consider the one-step estimator and the orthogonalized inactive score
\begin{align}
	\tilde\theta_M
	&\equiv\hat\theta_M+\hat H_M(\hat\theta_M)^{-1}
	S_M(\hat\theta_M),\label{eq: onestep}\\
	\tilde S_{-M}(\hat\theta_M)
	&\equiv S_{-M}(\hat\theta_M)
	-\hat H_{-M,M}(\hat\theta_M)\hat H_M(\hat\theta_M)^{-1}
	S_M(\hat\theta_M).\label{eq: effscore-M}
\end{align}
We collect these quantities together with the functional estimator in
\begin{equation}\label{eq: Zdef}
	Z\equiv n^{1/2}
	\begin{bmatrix}
		\tilde\rho_M\\
		\tilde\theta_M\\
		\tilde S_{-M}(\hat\theta_M)
	\end{bmatrix}
	\in\mathbb R^{d+p+1}.
\end{equation}
The construction thus treats a general scalar functional as the first
coordinate of the joint statistic, without requiring a separate
polyhedral representation.

\paragraph*{KKT conditions as polyhedral constraints.} Next we express the KKT conditions in  \eqref{eq:KKT1}-\eqref{eq:KKT2} as polyhedral constraints. 
Let $P\equiv[\bm0_{(|M|-d)\times d},I_{|M|-d}]$ extract the $\beta_M$
coordinates from $\theta_M$. Write $q\equiv|M|-d$ and
$r\equiv d+p-|M|$. Using $\|\mathtt u\|_\infty\leq1$, the closure
of the selection event $E_{M,\sm}\equiv
\{\hat M=M,\hsm=\sm\}$ has the polyhedral representation
\begin{equation}\label{eq: AZb}
	\{AZ\leq\hat b\},
\end{equation}
where
\begin{align}
	A&\equiv
	\begin{bmatrix}
		\bm0_{q}&-\mathrm{diag}(\sm)P&\bm0_{q\times r}\\
		\bm0_{r}&\bm0_{r\times|M|}&I_r\\
		\bm0_{r}&\bm0_{r\times|M|}&-I_r
	\end{bmatrix},\  
	\hat b\equiv n^{1/2}
	\begin{bmatrix}
		-\mathrm{diag}(\sm)P\hat H_M(\hat\theta_M)^{-1}v_M\\
		\lambda\bm1_r-\hat H_{-M,M}(\hat\theta_M)
		\hat H_M(\hat\theta_M)^{-1}v_M\\
		\lambda\bm1_r+\hat H_{-M,M}(\hat\theta_M)
		\hat H_M(\hat\theta_M)^{-1}v_M
	\end{bmatrix},\ 
	v_M\equiv
		\begin{bmatrix}
	\bm0_d\\
	\lambda\sm
		\end{bmatrix}.\label{def: Ab}
\end{align}
The zero first column of $A$ shows that including
$\tilde\rho_M$ in $Z$ does not alter the selection event.

Following \cite{Lee-Sun-Sun-Taylor(2016)}, Proposition 4.2,
we show that the active and inactive inequalities in
\eqref{eq: AZb} are equivalent to the corresponding KKT
conditions in \eqref{eq:KKT1}--\eqref{eq:KKT2}. First, the active constraints satisfy
\begin{align*}
	\left\{\mathrm{sign}(\hat\beta_M)=\sm\right\}
	&=
	\left\{\mathrm{diag}(\sm)P
	\left[
	\tilde\theta_M-
	\hat H_M(\hat\theta_M)^{-1}v_M
	\right]>0\right\}\\
	&=
	\left\{
	-\mathrm{diag}(\sm)P\tilde\theta_M
	<
	-\mathrm{diag}(\sm)P
	\hat H_M(\hat\theta_M)^{-1}v_M
	\right\}.
\end{align*}
To verify the first equality, note that, if 
$\mathrm{sign}(\hat\beta_M)=\sm$, by the restricted KKT conditions,
$S_M(\hat\theta_M)=v_M$, and hence
\eqref{eq: onestep} gives
$\tilde\theta_M-
\hat H_M(\hat\theta_M)^{-1}v_M
=\hat\theta_M$.
Thus, the sign restrictions imply the active inequalities. 
Conversely, suppose that the active inequalities hold.
By the restricted KKT conditions and \eqref{eq: onestep},
\begin{align*}
	0
	&\leq
	\left[
	\tilde\theta_M-
	\hat H_M(\hat\theta_M)^{-1}v_M
	-\hat\theta_M
	\right]'
	\left[v_M-S_M(\hat\theta_M)\right]\\
	&=
	-\left[S_M(\hat\theta_M)-v_M\right]'
	\hat H_M(\hat\theta_M)^{-1}
	\left[S_M(\hat\theta_M)-v_M\right]
	\leq0.
\end{align*}
The first inequality follows from the active sign
constraints and the coordinatewise restricted KKT
conditions. Since $\hat H_M(\hat\theta_M)$ is positive
definite, we obtain $S_M(\hat\theta_M)=v_M$.
Hence, by \eqref{eq: onestep}, the active inequalities
reduce to $\mathrm{diag}(\sm)P\hat\theta_M>0$,
or equivalently, $\mathrm{sign}(\hat\beta_M)=\sm$.
 Next, given $S_M(\hat\theta_M)=v_M$, the inactive
KKT condition \eqref{eq:KKT2} is equivalent to
\begin{align*}
	\left\{\|S_{-M}(\hat\theta_M)\|_\infty
	\leq\lambda\right\}
	&=
	\left\{
	-\lambda\bm1_r
	\leq
	\tilde S_{-M}(\hat\theta_M)
	+\hat H_{-M,M}(\hat\theta_M)
	\hat H_M(\hat\theta_M)^{-1}v_M
	\leq\lambda\bm1_r
	\right\}\\
	&=
	\left\{
	\begin{aligned}
		\tilde S_{-M}(\hat\theta_M)
		&\leq\lambda\bm1_r-
		\hat H_{-M,M}(\hat\theta_M)
		\hat H_M(\hat\theta_M)^{-1}v_M,\\
		-\tilde S_{-M}(\hat\theta_M)
		&\leq\lambda\bm1_r+
		\hat H_{-M,M}(\hat\theta_M)
		\hat H_M(\hat\theta_M)^{-1}v_M
	\end{aligned}
	\right\}.
\end{align*}
These are exactly the remaining inequalities in
\eqref{eq: AZb}.

We have therefore established both directions of the
equivalence: the active inequalities recover the
coefficient signs and active KKT conditions, while
the inactive inequalities recover the inactive KKT
conditions. Since the KKT conditions are necessary
and sufficient for the penalized optimization problem
and its solution is unique, the combined inequalities
characterize the selection event
$E_{M,\sm}$.
The exact characterization uses strict inequalities
for the active signs and weak inequalities for
inactive feasibility; \eqref{eq: AZb} uses the
corresponding boundary-inclusive representation.

\paragraph*{Gaussian approximation by Yurinskii coupling.}
For a generic model $M\in\mathcal M_n(m_n)$, define the population
full Hessian $H(\theta_M^*)\equiv\E[\hat H(\theta_M^*)]$, with blocks
$H_M\equiv[H(\theta_M^*)]_{M,M}$ and
$H_{-M,M}\equiv[H(\theta_M^*)]_{-M,M}$.
Define the full-dimensional score contribution
$s_{i,M}\equiv-w_ix_i\dot g(y_i,x_{iM}'\theta_M^*)$ and
$\xi_{i,M}\equiv s_{i,M}-\E[s_{i,M}]$. The covariance of the
centered score is
\begin{equation}\label{eq: population score covariance}
	I_M\equiv n^{-1}\sum_{i=1}^n\E[\xi_{i,M}\xi_{i,M}']
	=n^{-1}\sum_{i=1}^n\mathrm{Var}(s_{i,M}).
\end{equation}
For a functional with a sample-dependent counterpart,
suppose that
\begin{equation}\label{eq: functional averaging expansion}
	n^{1/2}\left(
	\hat\rho_M(\theta_M^*)-\rho_M(\theta_M^*)
	\right)
	=
	n^{-1/2}\sum_{i=1}^n\tau_{i,M}
	+
	o_p\left((\log 2p)^{-1/2}\right),
\end{equation}
where $\tau_{i,M}$ are independent, mean-zero
influence-function contributions. 
For a parameter functional, set $\tau_{i,M}=0$.
Define the augmented score covariance by
\begin{equation}\label{eq: augmented score covariance}
	\omega_{i,M}\equiv
	\begin{bmatrix}\xi_{i,M}\\\tau_{i,M}\end{bmatrix},
	\qquad
	\Omega_M\equiv n^{-1}\sum_{i=1}^n
	\E[\omega_{i,M}\omega_{i,M}'].
\end{equation}
The upper-left block of $\Omega_M$ is $I_M$.
Let $\hat\Omega_M$ be a symmetric positive-semidefinite estimator
of $\Omega_M$; its construction may use grouping variables or
replicate weights, as discussed in Section~\ref{subsec: SI aux}.

To estimate the covariance matrix of $Z$, define
\begin{equation}\label{eq: Sigma rho robust main}
	\hat\Sigma_M\equiv\hat Q_M\hat\Omega_M\hat Q_M',\quad
	\hat Q_M\equiv
	\begin{bmatrix}
		\nabla{\hat\rho}_M(\hat\theta_M)'\hat H_M(\hat\theta_M)^{-1}
		&\bm0_r'&1\\
		\hat H_M(\hat\theta_M)^{-1}&\bm0_{|M|\times r}&\bm0_{|M|}\\
		-\hat H_{-M,M}(\hat\theta_M)\hat H_M(\hat\theta_M)^{-1}
		&I_r&\bm0_r
	\end{bmatrix}.
\end{equation}
Its population counterpart is
\begin{equation}\label{eq: Qrho population main}
	\Sigma_M\equiv Q_M\Omega_MQ_M',\quad
	Q_M\equiv
	\begin{bmatrix}
		\nabla\rho_M(\theta_M^*)'H_M^{-1}&\bm0_r'&1\\
		H_M^{-1}&\bm0_{|M|\times r}&\bm0_{|M|}\\
		-H_{-M,M}H_M^{-1}&I_r&\bm0_r
	\end{bmatrix}.
\end{equation}
The restricted population first-order condition gives
$\E[S_M(\theta_M^*)]=0$, but it need not imply
$\nu_{-M}\equiv\E[S_{-M}(\theta_M^*)]=0$.
An omitted active regressor can leave a nonzero population score in
the excluded directions. If $\mathcal S_0\subseteq M$,
$\E[S(\theta_0)]=0$, and $\theta_M^*=\theta_{0,M}$, then
$\nu_{-M}=0$. Define the Gaussian mean and target variance by
\begin{equation}\label{eq: rho population quantities}
	\mu_M\equiv n^{1/2}
	\begin{bmatrix}
		\rho_M(\theta_M^*)\\\theta_M^*\\\nu_{-M}
	\end{bmatrix},\qquad
	\sigma_M^2\equiv\bm e_1'\Sigma_M\bm e_1.
\end{equation}
Partition $I_M$ conformably with $(M,-M)$ and let
$H_{M,-M}\equiv H_{-M,M}'$. By \eqref{eq: Qrho population main}, the
covariance between the active component
$H_M^{-1}n^{1/2}S_M(\theta_M^*)$ and the orthogonalized inactive score
$n^{1/2}\left(S_{-M}(\theta_M^*)-H_{-M,M}H_M^{-1}S_M(\theta_M^*)\right)$
is
\begin{equation*}
	H_M^{-1}\left([I_M]_{M,-M}-[I_M]_{M,M}H_M^{-1}H_{M,-M}\right).
\end{equation*}
Under the information matrix equality $I_M=H(\theta_M^*)$, which holds
in the unweighted GLM considered by \cite{Taylor-Tibshirani(2018)} when
the selected model contains the true support, this covariance is zero.
The active and inactive components are then asymptotically independent,
and, for a parameter functional, the inactive constraints can be dropped.
With survey weights, $I_M\neq H(\theta_M^*)$ in general
(Section~\ref{sec: model}), even when the selected model is correctly
specified; an underfitted selected model has the same effect. The
covariance then need not vanish, and the two components are
statistically dependent. We therefore account
for their joint covariance in the Gaussian approximation and condition on
both sets of constraints.
For a population AME, $\Omega_M$ additionally captures the
variance of the empirical average and its covariance with the
selection scores.

Under the rate conditions in Proposition~\ref{prop: SI validity},
Yurinskii coupling provides a Gaussian approximation to $Z$; see
Appendix~\ref{subsec: Proof SI validity}. Specifically, there exists
a Gaussian vector $T_M\sim N(0,\Sigma_M)$ such that
\begin{equation}\label{eq: AN}
	\|Z-\mu_M-T_M\|_\infty\conp0.
\end{equation}
We summarize this coupling as
\begin{equation}\label{eq: approx N}
	Z\stackrel a\sim N(\mu_M,\Sigma_M).
\end{equation}

\paragraph*{Truncated Gaussian pivot.}
Combined with the polyhedral constraint \eqref{eq: AZb}, the Gaussian
approximation \eqref{eq: approx N} leads to a truncated Gaussian pivot.
For $z\in\mathbb R^k$, a covariance matrix $\Sigma$, and a contrast
$\eta$ with $\eta'\Sigma\eta>0$, define
\begin{equation}\label{eq:cr}
	c(\Sigma,\eta)\equiv\frac{\Sigma\eta}{\eta'\Sigma\eta},
	\qquad
	r(z,\Sigma,\eta)\equiv
	\left(I_k-c(\Sigma,\eta)\eta'\right)z.
\end{equation}
For fixed $A,b$ and $c=c(\Sigma,\eta)$, define
\begin{align}
	\mathcal V^-(r)
	&\equiv\max_{j:(Ac)_j<0}
	\frac{b_j-(Ar)_j}{(Ac)_j},\label{eq:Vminus}\\
	\mathcal V^+(r)
	&\equiv\min_{j:(Ac)_j>0}
	\frac{b_j-(Ar)_j}{(Ac)_j},\label{eq:Vplus}\\
	\mathcal V^0(r)
	&\equiv\min_{j:(Ac)_j=0}\left(b_j-(Ar)_j\right).
	\label{eq:Vzero}
\end{align}
Here $(Ac)_j$ denotes the $j$th element of $Ac$; we use
$\max\varnothing=-\infty$ and $\min\varnothing=+\infty$.
The following result collects Lemma~5.1 and Theorem~5.2 of
\cite{Lee-Sun-Sun-Taylor(2016)}.

\begin{lemma}[Polyhedral lemma and truncated Gaussian pivot]
	\label{lem: Poly}
	Let $G\sim N(\mu,\Sigma)$ and let $A$, $b$, and $\eta$ be fixed,
	with $\eta'\Sigma\eta>0$. Write $c=c(\Sigma,\eta)$ and
	$r=r(G,\Sigma,\eta)$. Then $\eta'G$ is independent of $r$, and
	\begin{equation}\label{eq: PolyConst}
		\{AG\leq b\}
		=\{\mathcal V^-(r)\leq\eta'G\leq\mathcal V^+(r),\ 
		\mathcal V^0(r)\geq0\}.
	\end{equation}
	\phantomsection\label{lem: TGP}
	If $P(AG\leq b)>0$, then
	\begin{equation}\label{eq: TNP}
		F\left(\eta'G;\eta'\mu,\eta'\Sigma\eta,
		\mathcal V^-(r),\mathcal V^+(r)\right)
		\mid\{AG\leq b\}\sim U(0,1).
	\end{equation}
\end{lemma}

\paragraph*{Selective confidence interval.}
Set $c_M\equiv c(\Sigma_M,\bm e_1)$. In our application, take $\eta=\bm e_1$ and let
$\hat c_M\equiv c(\hat\Sigma_M,\bm e_1)$ and
$\hat r\equiv r(Z,\hat\Sigma_M,\bm e_1)$.
The truncation limits $\mathcal V^-(\hat r)$ and
$\mathcal V^+(\hat r)$ are computed using $A$, $\hat b$, and
$\hat c_M$ in \eqref{eq:Vminus}--\eqref{eq:Vplus}.
For $\zeta\in(0,1)$, define the level $1-\zeta$ SI confidence
interval $\mathrm{CI}_{\hat M}\equiv
[\tilde q_l,\tilde q_u]$ for $\rho_M(\theta_M^*)$,
where the lower and upper endpoints solve
\begin{align}
	F\left(n^{1/2}\tilde\rho_M;n^{1/2}\tilde q_l,
	\bm e_1'\hat\Sigma_M\bm e_1,
	\mathcal V^-(\hat r),\mathcal V^+(\hat r)\right)
	&=1-\frac\zeta2,\label{eq: SI CI3}\\
	F\left(n^{1/2}\tilde\rho_M;n^{1/2}\tilde q_u,
	\bm e_1'\hat\Sigma_M\bm e_1,
	\mathcal V^-(\hat r),\mathcal V^+(\hat r)\right)
	&=\frac\zeta2.\label{eq: SI CI4}
\end{align}
Although $A$ is fixed for a given model and sign vector, $\hat b$ and
$\hat\Sigma_M$ are random and $Z$ is only approximately Gaussian.
Lemma~\ref{lem: Poly} therefore does not by itself establish validity
of the feasible sample pivot. The following result supplies the needed
rate, stability, and selection probability conditions.

\subsection{Asymptotic validity}\label{subsec: SI validity}
We now state the main results establishing the asymptotic validity of selective inference for a general parameter functional. The regularity conditions required for these results are collected in Assumptions \ref{A: AsyValid}, \ref{A: score covariance} and \ref{A: SI pivot} below.
\begin{proposition}[Asymptotic validity of SI]
	\label{prop: SI validity}
	Fix a deterministic sequence
	$M=M_n\in\mathcal M_n(m_n)$, a sign vector
	$\sm$, and a functional $\rho_M$.
	Let \eqref{eq: functional averaging expansion} and Assumptions~\ref{A: AsyValid},
	\ref{A: score covariance}, and \ref{A: SI pivot}
	hold, and suppose that
	$\lambda=C\sqrt{(\log 2p)/n}$ for some fixed
	$0<C<\infty$. If
	\begin{equation*}
		\frac{p^2(\log 2p)^6}{n}\to0,\quad
		m_n\log 2p\sqrt{\frac{p}{n}}\to0,
		\quad\text{and}\quad
		\frac{m_n^3(\log 2p)^3}{n}\to0,
	\end{equation*}
	then $\mathrm{CI}_{\hat M}$, defined by the
	endpoints solving
	\eqref{eq: SI CI3}--\eqref{eq: SI CI4},
	satisfies \eqref{eq: SI objective}
	for every $\zeta\in(0,1)$.
\end{proposition}
\begin{corollary}[Model-conditional validity]
\label{cor: SI model validity}
Fix a deterministic sequence $M=M_n\in\mathcal M_n(m_n)$ and a scalar functional $\rho_M$.
Suppose \eqref{eq: functional averaging expansion} and Assumptions~\ref{A: AsyValid}, \ref{A: score covariance}, and \ref{A: SI pivot}\ref{SIpivot poly} hold, together with the tuning parameter and rate
conditions of Proposition~\ref{prop: SI validity}.
If
\begin{equation*}
\sup_n(|M|-d)<\infty,
\qquad
\liminf_{n\to\infty}P(\hat M=M)>0,
\end{equation*}
then the confidence interval $\mathrm{CI}_{\hat M}$ defined by \eqref{eq: SI CI3}--\eqref{eq: SI CI4} satisfies \eqref{eq: SI model objective}.
\end{corollary}
The proofs of Proposition~\ref{prop: SI validity} and Corollary~\ref{cor: SI model validity} are given in Appendix~\ref{subsec: Proof SI validity} and Appendix~\ref{subsec: Proof SI model validity}, respectively.

The assumption of fixed $p$ is commonly used in the SI literature \citep{Lee-Sun-Sun-Taylor(2016),Tian-Taylor(2017),Taylor-Tibshirani(2018),KKK2022}, while \cite{Taylor-Tibshirani(2018)} provide a heuristic argument for standard GLMs. Our main result establishes asymptotic validity for survey-weighted GLMs with i.n.i.d. observations and a growing number of covariates. Survey weighting and heteroskedasticity may induce dependence between the random quantities defining the active and inactive KKT constraints, which are therefore treated jointly. The result also accommodates inference on pseudo-true parameters of possibly misspecified selected models and nonlinear functionals, including population AMEs.

Corollary~\ref{cor: SI model validity} establishes that the same SI confidence intervals attain asymptotic coverage conditional on the selected model alone. The result follows by averaging over the possible sign configurations, provided the number of selected penalized regressors remains bounded and the model-selection probability does not vanish. Thus, model-conditional validity is achieved without modifying the inference procedure, while allowing the number of candidate regressors $p$ to grow with the sample size. When the number of selected penalized regressors grows, the averaging argument would require uniform control over an increasing number of sign configurations, and we do not claim model-conditional validity in that case.

The rate conditions are sufficient and may be conservative. In particular, the first condition arises from the Gaussian coupling of the $(d+p+1)$-dimensional augmented score vector. For bounded $m_n$, it is the most restrictive condition and requires $p=o(\sqrt{n}/(\log 2p)^3)$. Under i.i.d. sampling, sharper Gaussian approximations tailored to the polyhedral selection event could potentially relax this restriction. The present theory does not cover $p\geq n$. Extending SI to that regime would require a different treatment of the high-dimensional Hessian and covariance matrices and of the growing polyhedral selection event, which we leave for future research.

We next provide the regularity conditions underlying  Proposition~\ref{prop: SI validity} and Corollary~\ref{cor: SI model validity}.
The first set of assumptions imposes regularity conditions at the selected-model population target $\theta_M^*$ rather than at the full-model parameter.

\begin{assumption}[Sampling and model regularity]\label{A: AsyValid}
	For every deterministic sequence $M=M_n\in\mathcal M_n(m_n)$,
	the following conditions hold:
	\begin{enumerate}[label={\rm(\alph*)}]
		\item\label{AsyValid sampling}
	$\{(y_i,x_i',w_i)'\}_{i=1}^n$ are independent
		across $i$. For some constants $0<C_l<C_u<\infty$, 
		$\max_{1\le i\le n}\|x_i\|_{\psi_2}\le C_u$,
		$\max_{1\le i\le n}\|x_i\|_\infty\le C_u$ a.s.,
		and $C_l\le w_i\le C_u$ a.s. for all $i$.

		\item\label{AsyValid convex}\label{SIlocal derivatives}
		$g(y,t)$ is convex and twice continuously
		differentiable in $t\in\mathbb R$ for every $y\in\mathcal Y$.
		There are constants $\bar\eta,L_0,C_g>0$, independent of $M$
		and $n$, such that the following conditions hold uniformly
		over $i$, $y\in\mathcal Y$, and $t,\tilde t$ satisfying
		$\max\{
		|t-x_{iM}'\theta_M^*|,
		|\tilde t-x_{iM}'\theta_M^*|
		\}\le\bar\eta$:
		\begin{equation*}
			|\dot g(y,x_{iM}'\theta_M^*)|\le C_g,\quad
			|\ddot g(y,t)|\le C_g,\quad
			|\ddot g(y,t)-\ddot g(y,\tilde t)|\le L_0|t-\tilde t|.
		\end{equation*}

		\item\label{SIlocal target}
		$\mathcal E_M$ in \eqref{def: proj} has a
		unique interior minimizer $\theta_M^*$, and
		$\max_i|x_{iM}'\theta_M^*|\le C_u$ a.s.

		\item\label{SIlocal curvature}\label{SIlocal variance}
		There are constants $0<c_H<C_H<\infty$ and
		$0<c_I<C_I<\infty$, independent of $M$ and $n$, such that
		$\lambda_{\min}(H_M)\ge c_H$,
		$\|H(\theta_M^*)\|_2\le C_H$, and
		$c_I\le\lambda_{\min}(I_M)
		\le\lambda_{\max}(I_M)\le C_I$.

\item\label{SIpivot rho}
The functional $\rho_M$ is continuously differentiable
in a neighborhood of $\theta_M^*$, with a locally
Lipschitz Jacobian satisfying $
0<c_\rho\leq
\|\nabla\rho_M(\theta_M^*)\|_2^2
\leq C_\rho<\infty.$

The sample functional $\hat\rho_M$ has a
Lipschitz Jacobian on a fixed neighborhood
of $\theta_M^*$, with Lipschitz constant
$O_p(1)$,
and \eqref{eq: functional averaging expansion}
holds with independent, mean-zero $\tau_{i,M}$
having uniformly bounded sub-Gaussian norms.
Moreover, $
\|\nabla{\hat\rho}_M(\theta_M^*)-
\nabla\rho_M(\theta_M^*)\|_2
=O_p\left(\sqrt{\frac{m_n}{n}}\right),$
and $\inf_n\sigma_M^2>0$.
\end{enumerate}
\end{assumption}

Assumption~\ref{A: AsyValid}\ref{AsyValid sampling}
allows independent, non-identically distributed observations,
including the stratified sampling designs in
Section~\ref{sec: model}.
Survey weights may be random but must be uniformly bounded
above and away from zero. The weights need not sum to one.
The \texttt{R} package \texttt{glmnet} rescales observation
weights to sum to $n$, and we adopt this mean-one convention
in the numerical implementation.

The remaining conditions impose at $\theta_M^*$ analogues
of the curvature, smoothness, and variance conditions
commonly imposed at $\theta_0$ in analyses of regularized
$M$ estimators and GLMs
\citep{Negahban-etal(2012),vandeGeer-etal(2014),xia2023}.

In particular, Assumption~\ref{A: AsyValid}\ref{SIpivot rho}
requires a smooth, nondegenerate functional and controls
the empirical approximation of its value and derivative.
The influence-function expansion
\eqref{eq: functional averaging expansion}
is immediate for parameter functionals, for which
the empirical averaging term vanishes.
In Section~\ref{subsec: SI aux}, we verify this
expansion for the population AME and describe the
corresponding covariance estimation.

\begin{assumption}[Joint covariance estimation]
	\label{A: score covariance}
	For every deterministic sequence
	$M=M_n\in\mathcal M_n(m_n)$, there exists a deterministic sequence
	$r_{I,n}\to0$ such that
	\begin{equation}\label{eq: I general rate}
		\|\hat\Omega_M-\Omega_M\|_2=O_p(r_{I,n}),\qquad
		(\log 2p)^{1/2}r_{I,n}\to0.
	\end{equation}
\end{assumption}

Assumption~\ref{A: score covariance} accommodates
survey-specific covariance estimators.
For a parameter functional, $\tau_{i,M}=0$,
so only the score covariance $I_M$ needs to be estimated.
For a population AME, the joint covariance also accounts
for the variability of empirical averaging and its
covariance with the pseudo-score.

\begin{assumption}[Selection-event regularity]
	\label{A: SI pivot}~
	\begin{enumerate}[label={\rm(\alph*)}]
		\item\label{SIpivot poly}
		For some $\kappa>0$, the nonzero polyhedral slopes satisfy
		$\min_{j:(A c_M)_j\neq0}|(A c_M)_j|\geq\kappa$
		for all sufficiently large $n$.
		
		\item\label{SIpivot selection}
		The selection event satisfies
		$\liminf_{n\to\infty}P[E_{M,\sm}]>0$.
	\end{enumerate}
\end{assumption}

Assumption~\ref{A: SI pivot} imposes the stability
conditions needed for conditional inference.
Part~\ref{SIpivot poly} keeps the nonzero polyhedral
slopes away from zero, while Part~\ref{SIpivot selection}
ensures that approximation errors remain negligible
after conditioning on selection.
The growth rates in Proposition~\ref{prop: SI validity}
control the Gaussian coupling, linearization, covariance,
and threshold approximation errors.
The nonvanishing selection probability condition
excludes asymptotically rare selection events,
including those associated with underfitted models
whose inactive population scores remain bounded
away from zero as $\lambda\to0$.

\subsection{Covariance estimation and auxiliary results}
\label{subsec: SI aux}

We first verify the influence-function and Jacobian
conditions in \eqref{eq: functional averaging expansion} and Assumption~\ref{A: AsyValid}\ref{SIpivot rho}
for the population AME in the survey-weighted logit
model. We then describe estimators of the score and
augmented covariance matrices and establish the
sample matrix rates used in the proof of
Proposition~\ref{prop: SI validity}.

\paragraph*{AME linearization and Jacobian convergence.}
For a population AME, the additional influence-function
contribution in \eqref{eq: augmented score covariance}
is
\begin{equation}\label{eq: AME influence term}
\tau_{i,M}
\equiv
\frac{w_i\left(\mathrm{ME}_{ij}(\theta_M^*)
-\rho_M(\theta_M^*)\right)
-\E\left[w_i\left(\mathrm{ME}_{ij}(\theta_M^*)-\rho_M(\theta_M^*)\right)\right]
}{n^{-1}\sum_{k=1}^n\E[w_k]}.
\end{equation}
To verify \eqref{eq: functional averaging expansion}, recall that
$\hat\rho_M(\theta_M^*)
=(\sum_{i=1}^n w_i\mathrm{ME}_{ij}(\theta_M^*))/
(\sum_{i=1}^n w_i)$.
Under Assumption~\ref{A: AsyValid}\ref{AsyValid sampling},
the observations are independent and the weights
are bounded above and away from zero.

The individual marginal effects are uniformly bounded
in the weighted logit model. For binary regressors,
this follows from the boundedness of $\Lambda$.
For continuous regressors,
$\mathrm{ME}_{ij}(\theta_M)
=\beta_{M,j}\Lambda'(x_{iM}'\theta_M)$.
The uniform bound
$\lambda_{\min}(H_M)\geq c_H>0$ implies
$c_H\|\theta_M^*\|_2^2
\leq\theta_M^{*\prime}H_M\theta_M^*
\leq C_u\sup_{t\in\mathbb R}t^2\Lambda'(t)<\infty$.
Hence $\|\theta_M^*\|_2=O(1)$ uniformly over
$M\in\mathcal M_n(m_n)$, and the marginal effects
are bounded in this case as well.

A second-order ratio expansion therefore shows that
the expectation of the sample ratio differs from the
ratio of expectations by $O(n^{-1})$.
Since $\rho_M(\theta_M^*)=\E[\hat\rho_M(\theta_M^*)]$,
a first-order ratio expansion yields
\begin{equation}\label{eq: AME sample linearization}
	n^{1/2}\left(
	\hat\rho_M(\theta_M^*)-\rho_M(\theta_M^*)
	\right)
	=
	n^{-1/2}\sum_{i=1}^n\tau_{i,M}
	+O_p(n^{-1/2}).
\end{equation}
Since $(\log 2p)/n\to0$ follows from
$p^2/n\to0$ in
Proposition~\ref{prop: SI validity},
the remainder is
$o_p((\log 2p)^{-1/2})$, verifying
\eqref{eq: functional averaging expansion}.

We next verify the Jacobian convergence rate in
Assumption~\ref{A: AsyValid}\ref{SIpivot rho}.
Differentiation of the sample AME gives
$\nabla\hat\rho_M(\theta_M^*)
=(\sum_{i=1}^n w_i
\nabla\mathrm{ME}_{ij}(\theta_M^*))/
(\sum_{i=1}^n w_i)$.
By differentiation under the expectation,
$\nabla\rho_M(\theta_M^*)
=\E[\nabla\hat\rho_M(\theta_M^*)]$. The elements of
$\nabla\mathrm{ME}_{ij}(\theta_M^*)$
are uniformly bounded. For binary regressors,
this follows from bounded regressors and
$|\Lambda'(t)|\leq1/4$. For continuous
regressors, differentiation of
$\beta_{M,j}\Lambda'(x_{iM}'\theta_M)$
involves $\Lambda'$, $\Lambda''$, and
$\beta_{M,j}$, which are uniformly bounded
at $\theta_M^*$ under the conditions
established above.

By independence and the uniform boundedness of the
weights and marginal-effect Jacobians,
\begin{equation*}
\E\left[\left\|
n^{-1}\sum_{i=1}^n
\left\{
w_i\nabla{\mathrm{ME}}_{ij}(\theta_M^*)
-\E[w_i\nabla{\mathrm{ME}}_{ij}(\theta_M^*)]
\right\}
\right\|_2^2\right]
=O(|M|/n),
\end{equation*}
and
$\E\left[|n^{-1}\sum_{i=1}^n(w_i-\E[w_i])|^2\right]
=O(n^{-1})$.
Since $n^{-1}\sum_i w_i\geq C_l$ and
$\|n^{-1}\sum_i\E[w_i
\nabla{\mathrm{ME}}_{ij}(\theta_M^*)]\|_2
=O(\sqrt{|M|})$,
a direct expansion of the sample ratio gives
\begin{equation*}
\E\left[\left\|
\nabla{\hat\rho}_M(\theta_M^*)
-
\frac{
	\sum_i\E[w_i\nabla{\mathrm{ME}}_{ij}(\theta_M^*)]
}{
	\sum_i\E[w_i]
}
\right\|_2^2\right]
=O(|M|/n).
\end{equation*}
As $\nabla\rho_M(\theta_M^*)
=\E[\nabla{\hat\rho}_M(\theta_M^*)]$,
centering at this expectation cannot increase
the mean squared deviation. Hence
$\E\left[\|\nabla{\hat\rho}_M(\theta_M^*)
-\nabla\rho_M(\theta_M^*)\|_2^2\right]
=O(|M|/n)=O(m_n/n)$,
which by Markov's inequality implies
$\|\nabla{\hat\rho}_M(\theta_M^*)
-\nabla\rho_M(\theta_M^*)\|_2
=O_p(\sqrt{m_n/n})$,
as required by
Assumption~\ref{A: AsyValid}\ref{SIpivot rho}.

The local Lipschitz conditions in
Assumption~\ref{A: AsyValid}\ref{SIpivot rho}
also hold for the weighted logit AME.
Since $\Lambda'$, $\Lambda''$, and $\Lambda'''$
are uniformly bounded and
$\|\theta_M^*\|_2=O(1)$, the Hessian of the
sample functional is bounded on any fixed
neighborhood of $\theta_M^*$ by
\begin{equation*}
	\sup_{\|\theta_M-\theta_M^*\|_2\leq\varepsilon}
	\|\nabla^2\hat\rho_M(\theta_M)\|_2
	\leq
	C\left(1+\|n^{-1}X'X\|_2\right)
	=O_p(1),
\end{equation*}
where the last equality follows from
\eqref{eq:design-matrix-bound}.
For binary regressors, the same bound holds
because setting one bounded regressor to zero
or one preserves the order of the sample
design-matrix bound.
The corresponding population Jacobian is
locally Lipschitz by the same derivative
bounds and the uniform population
second-moment bound.

\paragraph*{Covariance estimation.}
To estimate the score covariance $I_M$, suppose that
observations can be partitioned into groups
$\mathcal I_h$, $h=1,\ldots,J$, within which the
population score mean is common. Let
$n_h=|\mathcal I_h|\geq2$ and define
\begin{align}
	\hat s_{i,M}
	&\equiv
	-w_i x_i\dot g(y_i,x_{iM}'\hat\theta_M),
	\quad \bar s_{h,M}
	\equiv
	n_h^{-1}\sum_{i\in\mathcal I_h}\hat s_{i,M},
	\nonumber\\
	\hat I_M^{\mathrm{grp}}
	&\equiv
	\frac{1}{n}\sum_{h=1}^J
	\frac{n_h}{n_h-1}
	\sum_{i\in\mathcal I_h}
	(\hat s_{i,M}-\bar s_{h,M})
	(\hat s_{i,M}-\bar s_{h,M})'.
	\label{eq: grouped pseudo covariance}
\end{align}
The factor $n_h/(n_h-1)$ is the usual
degrees-of-freedom correction. Group centering
accounts for nonzero, group-specific score means
without requiring identical distributions within
groups. In particular, it accommodates i.n.i.d.
observations whose score means need not vanish
individually, although their average is zero in
the active coordinates. This centering does not
change the target parameter $\theta_M^*$.

When $\E[s_{i,M}]=0$ for every $i$, no mean
correction is needed, and the outer-product
estimator $n^{-1}\sum_i\hat s_{i,M}\hat s_{i,M}'$
targets $I_M$. For a correctly specified logit model with survey weights fixed or measurable with respect to the regressors, the conditional information identity also yields
\begin{equation*}
\hat I_M^{\mathrm{logit}}
=
n^{-1}\sum_{i=1}^n
w_i^2x_ix_i'
\hat p_{i,M}(1-\hat p_{i,M}),
\qquad
\hat p_{i,M}=\Lambda(x_{iM}'\hat\theta_M).
\end{equation*}
This form is used in the simulations, but not
in the empirical application, where the selected
model may be misspecified.

For the population AME, a feasible influence-function
contribution is
\begin{equation*}
	\hat\tau_{i,M}
	\equiv
	\frac{
		w_i\left(\mathrm{ME}_{ij}(\hat\theta_M)
		-\hat\rho_M(\hat\theta_M)\right)
	}{
		n^{-1}\sum_{k=1}^n w_k
	}.
\end{equation*}
A grouped estimator of the augmented covariance
$\Omega_M$ is obtained by replacing
$\hat s_{i,M}$ in
\eqref{eq: grouped pseudo covariance} with
$\hat\omega_{i,M}
=(\hat s_{i,M}',\hat\tau_{i,M})'$.
Its consistency under a common within-group
mean condition is established in
Lemma~\ref{lem:pseudo-score-covariance}.

\paragraph*{Auxiliary matrix results.}
We next establish the sample matrix rates used in
the proof, for the design matrices $X$ and $X_M$
defined in Section~\ref{sec: model}.

\begin{lemma}[Sample matrix concentration]
	\label{lem:sample-matrix-consistency}
	Suppose that
	Assumption~\ref{A: AsyValid}\ref{AsyValid sampling}--%
	\ref{AsyValid convex} holds.
	For every deterministic sequence
	$M=M_n\in\mathcal M_n(m_n)$, if $p/n\to0$,
	then
	\begin{align}
		\|\hat H(\theta_M^*)-H(\theta_M^*)\|_2
		&=
		O_p\!\left(\sqrt{\frac pn}\right),
		\label{eq:Hessian-concentration}\\
		\|n^{-1}X'X\|_2
		&=O_p(1).
		\label{eq:design-matrix-bound}
	\end{align}
\end{lemma}

The proof is given in
Appendix~\ref{proof:sample-matrix-consistency},
using covariance concentration results for
sub-Gaussian vectors
\citep{Vershynin(2010),Vershynin(2018)}.

\begin{lemma}[Grouped covariance estimation]
	\label{lem:pseudo-score-covariance}
	Suppose Assumption~\ref{A: AsyValid} holds,
	and let $M=M_n\in\mathcal M_n(m_n)$ be deterministic.
	Suppose the observations admit a partition
	$\{\mathcal I_h\}_{h=1}^J$, with $J$ fixed, such that
	\begin{equation}\label{eq: common group score mean}
		\E[s_{i,M}]=\mu_{h,M},
		\qquad i\in\mathcal I_h,
	\end{equation}
	and, for constants $c_s,c_\lambda>0$,
	$\min_h n_h/n\geq c_s$ and
	$\lambda\geq c_\lambda n^{-1/2}$
	for all sufficiently large $n$.
	If $p/n\to0$ and
	$m_n(n^{-1/2}+\lambda)\to0$, then
	\begin{equation}\label{eq:I-consistency}
		\|\hat I_M^{\mathrm{grp}}-I_M\|_2
		=O_p\!\left(\sqrt{\frac pn}+m_n\lambda\right).
	\end{equation}
	For a population AME, if additionally
	$\E[w_i\{\mathrm{ME}_{ij}(\theta_M^*)
	-\rho_M(\theta_M^*)\}]$
	is constant within each group, then
	\begin{equation}\label{eq:Omega-consistency}
		\|\hat\Omega_M^{\mathrm{grp}}-\Omega_M\|_2
		=O_p\!\left(\sqrt{\frac pn}+m_n\lambda\right).
	\end{equation}
\end{lemma}
The proof is given in
Appendix~\ref{proof:pseudo-score-covariance}.
Moreover, $\sqrt{\log2p}(\sqrt{p/n}+m_n\lambda)=o(1)$ under the rate
conditions of Proposition~\ref{prop: SI validity}.
Thus, the grouped estimators verify
Assumption~\ref{A: score covariance}
for parameter functionals and, under the additional
group-mean condition, for population AMEs.

\section{Simulations}\label{sec: simul}

This section presents simulation evidence on the performance
of the proposed selective inference procedure.
We consider a logit model under two sampling schemes:
standard stratified sampling and exogenous stratification.
The empirical size of the tests is examined for different
sample sizes and numbers of regressors.

\paragraph*{Data-generating process.}
For each scheme, we consider four strata with population
proportions $q=(0.1,0.2,0.3,0.4)'$.
An equal number of observations is drawn from each stratum,
so that $n_h=n/4$, $h=1,\ldots,4$.
The corresponding survey weights are $w_i=4q_{H_i}$,
where $H_i$ denotes the stratum, and have sample mean one.

We set $d=1$ and write
$x_i=(1,\tilde{x}_i')'$, where
$\tilde{x}_i=(\tilde{x}_{i1},\ldots,\tilde{x}_{ip})'$.
The first non-constant regressor is binary, with
$\tilde{x}_{i1}\mid\{H_i=h\}\sim
\mathrm{Bernoulli}(\pi_h)$.
Under standard stratified sampling,
$\pi=(0.5,0.5,0.5,0.5)'$,
whereas under exogenous stratification,
$\pi=(0.80,0.65,0.50,0.35)'$.
In both cases, $\sum_{h=1}^4q_h\pi_h=0.5$,
so that the population proportion of observations
with $\tilde{x}_{i1}=1$ is the same across the two designs.

The remaining regressors are generated as follows:
\begin{equation}\label{eq:bounded nuisance dgp}
	\tilde{x}_{ij}=\frac{\sqrt{3}}{2}
	\left(2\Phi(Z_{ij})-1\right),
	\qquad j=2,\ldots,p,
\end{equation}
where $Z_i=(Z_{i2},\ldots,Z_{ip})'$ is a mean-zero
Gaussian vector with unit marginal variances and
$\mathrm{Corr}(Z_{ij},Z_{ik})=0.3^{|j-k|}$.
Thus, in the population,
$\tilde{x}_{ij}\sim U(-\sqrt{3}/2,\sqrt{3}/2)$,
with mean zero and variance $1/4$, equal to the
population variance of $\tilde{x}_{i1}$.
All penalized regressors therefore have comparable
population scales and are uniformly bounded by one.

Under standard stratified sampling, stratum membership
is independent of the regressors. Under exogenous
stratification, the strata are determined by
$\tilde{x}_{i2}$. Specifically, an observation belongs
to stratum $h$ if
\begin{equation*}
\sum_{\ell=1}^{h-1}q_\ell
<\Phi(Z_{i2})\leq
\sum_{\ell=1}^{h}q_\ell,\qquad h=1,\ldots,4,
\end{equation*}
where the empty sum is zero. Observations are generated within each stratum by
drawing $Z_{i2}$ from the corresponding truncated
normal distribution and generating the remaining
Gaussian coordinates conditionally according to
the specified correlation structure.
Thus, stratum membership is a deterministic function
of the regressors, as in Example~\ref{ex: ES}.

Conditional on the regressors, the dependent variable
is generated according to
\begin{equation}\label{eq:sim dgp bounded}
	y_i\,\vert\, x_i\sim
	\mathrm{Bernoulli}\left(
	\Lambda(-0.5+\tilde{x}_{i1})
	\right).
\end{equation}
The true parameter is therefore
$\theta_0=(-0.5,1,\bm 0_{p-1}')'$.
The corresponding average marginal effect is
$\mathrm{AME}_2=\Lambda(0.5)-\Lambda(-0.5)
\approx0.2449$.
Since the intercept is always included and
$\tilde{x}_{i1}$ is the only regressor with a
nonzero slope coefficient, any selected model
containing $\tilde{x}_{i1}$ also contains the
true support, and hence
$\theta_M^*=\theta_{0,M}$.

We consider sample sizes $n\in\{200,400,800\}$ and
vary the number of regressors such that
$p/n\in\{0.01,0.025,0.05,0.10,0.25\}$.\footnote{
	Sequences with $p/n$ fixed at a positive value
	do not satisfy the growth condition of
	Proposition~\ref{prop: SI validity}.}
The empirical size of the tests is evaluated at the
$5\%$ significance level for the null hypotheses
\begin{equation}\label{eq: H0 sim}
	H_0:\theta_{(2)}=1,
	\qquad
	H_0:\mathrm{AME}_2=\Lambda(0.5)-\Lambda(-0.5),
\end{equation}
using $3{,}000$ simulation replications for each
combination of $n$, $p$, and sampling design.

\paragraph*{Tuning parameters.}
Proposition~\ref{prop: SI validity} is stated for a
deterministic tuning sequence. To separate tuning parameter
selection from inference, we calibrate the tuning parameter
using independent auxiliary samples.
For each $(n,p,\text{design})$ combination, we generate
49 independent auxiliary samples. On each auxiliary sample,
we compute (i) the 10-fold cross-validation (CV) penalty
minimizing out-of-fold binomial deviance and (ii) the
bootstrapping-after-cross-validation (BCV) penalty of
\cite{ChetverikovSorensen2025}.
The medians of the 49 resulting CV and BCV penalties,
denoted by $\lambda_{\mathrm{CV}}$ and
$\lambda_{\mathrm{BCV}}$, respectively, are held fixed
throughout the $3{,}000$ Monte Carlo replications for the
corresponding combination.
We refer to these procedures as SI--fixed CV and
SI--fixed BCV.

For comparison, we also consider SI--adaptive CV and
SI--adaptive BCV, where the tuning parameter is selected
separately in each Monte Carlo replication using the same
sample on which inference is conducted.
For each simulated sample, we select the penalty
$\hat\lambda$ using either 10-fold CV or BCV, estimate
the corresponding $\ell_1$-PML model, and apply the
proposed SI procedure to the selected model.
Adaptive CV and BCV use the same random 10-fold
partition within each replication.
These procedures reflect the common practice of selecting
the tuning parameter from the observed data.
Although the data-dependent tuning step is not covered by
Proposition~\ref{prop: SI validity}, we include these
procedures to examine their finite-sample performance.

For BCV, the score residuals incorporate the survey weights,
and the multiplier bootstrap uses 500 draws with $c_0=1.1$
and $\alpha_{\mathrm{BCV}}=0.1/\log\{\max(p,n)\}$.
Because all non-constant regressors have comparable
population variances, no sample standardization is used
in the simulations.

\paragraph*{PoSI1.}
For comparison, we implement the restricted PoSI procedure of
\cite{Berk2013}, denoted PoSI1. The prespecified model family is
\begin{equation}\label{eq: posi M}
	\mathcal{M}^{\mathrm{PoSI1}}
	=\left\{
	M\subseteq\{1,\ldots,p+1\}:
	1\in M,\ 1\leq |M|-1\leq 2
	\right\},
\end{equation}
which contains all models with an intercept and either one or
two non-constant regressors.
For the target regressor $\tilde{x}_{i1}$, corresponding
to the second coordinate of $x_i$, the calibration is
restricted to
$\mathcal{M}^{\mathrm{PoSI1}}_{2\cdot}
=\{M\in\mathcal{M}^{\mathrm{PoSI1}}:2\in M\}$.

The PoSI1 critical value is computed from an independent pilot
sample using $B=999$ stratified bootstrap replications, once for
each $(n,p,\text{design})$ combination, and is then held fixed
across the Monte Carlo replications.
The standard survey-weighted $t$-test from the full model,
$t_{\mathrm{svy}}$, is also reported as a benchmark that ignores
the selection step.

\paragraph*{Results.}
Tables~\ref{tab:Level0} and~\ref{tab:LevelES} report the
empirical sizes of the tests under standard stratified sampling
and exogenous stratification, respectively. The results under
both sampling schemes are qualitatively similar. The SI
procedures show reasonable size control for both hypotheses
across different sample sizes and numbers of regressors.
In particular, SI--fixed CV and SI--fixed BCV yield rejection
frequencies generally close to the nominal $5\%$ level.

The adaptive procedures also perform well, with results
comparable to those obtained with fixed tuning. Some
overrejection is observed at $n=200$, but the empirical
sizes are generally close to the nominal level for
$n=400$ and $n=800$.

PoSI1 tends to underreject, particularly as $p/n$ increases.
In contrast, the rejection rate of the standard
survey-weighted $t$-test, denoted by $t_{\mathrm{svy}}$,
deteriorates as $p/n$ grows. At $p/n=0.25$ and $n=800$,
its coefficient rejection frequencies reach $58.9\%$
and $54.3\%$ under standard and exogenous stratification,
respectively. The size distortions for the AME test
are less pronounced. Finally, BCV selects the regressor
of interest less frequently in smaller samples, while
at $n=800$ it is selected in nearly every replication
under all four SI procedures.

\begin{table}[htbp]
 \caption{Empirical rejection frequencies (\%) of the tests for
 $H_0:\theta_{(2)}=1$ and $H_0:\mathrm{AME}_2=0.2449$ at the $5\%$ level.
 Standard stratified sampling, bounded-continuous design.}
 \label{tab:Level0}
 \begin{center}
 \footnotesize
 \begin{tabular}{lcccccccccc}
 \toprule
 & \multicolumn{5}{c}{$H_0:\theta_{(2)}=1$}
 & \multicolumn{5}{c}{$H_0:\mathrm{AME}_2=0.2449$}\\
 \cmidrule(lr){2-6}\cmidrule(lr){7-11}
 & \multicolumn{5}{c}{$p/n$} & \multicolumn{5}{c}{$p/n$}\\
 \cmidrule(lr){2-6}\cmidrule(lr){7-11}
 & .01 & .025 & .05 & .10 & .25 & .01 & .025 & .05 & .10 & .25\\
 \midrule
 \multicolumn{11}{c}{$n=200$}\\
 \midrule
 SI--fixed CV  &4.8&5.3&5.4&5.4&6.1 &4.5&4.9&5.0&4.8&5.5\\
 SI--fixed BCV &4.9&5.4&5.4&4.7&5.0 &5.1&5.4&5.4&4.4&5.4\\
 SI--adaptive CV &4.4&4.8&5.6&6.7&7.4 &4.3&4.5&5.3&6.0&6.6\\
 SI--adaptive BCV &4.3&5.9&6.6&6.9&6.8 &4.3&5.9&6.4&6.8&6.6\\
 PoSI1           &3.0&3.0&1.9&2.0&1.0 &3.6&3.4&2.3&2.2&1.4\\
 $t_{\mathrm{svy}}$ &4.3&4.9&6.1&8.2&25.1 &4.9&5.2&6.2&7.0&15.0\\
 $n_{\mathrm{sel}}$ fixed CV  &2988&2968&2919&2876&2764 &2988&2968&2919&2876&2764\\
 $n_{\mathrm{sel}}$ fixed BCV &1006&885&811&767&483 &1006&885&811&767&483\\
 $n_{\mathrm{sel}}$ adaptive CV  &2905&2814&2702&2616&2397 &2905&2814&2702&2616&2397\\
 $n_{\mathrm{sel}}$ adaptive BCV &1028&973&834&765&575 &1028&973&834&765&575\\
 \midrule
 \multicolumn{11}{c}{$n=400$}\\
 \midrule
 SI--fixed CV  &5.4&5.4&5.6&5.0&5.3 &5.1&5.2&5.2&4.7&5.1\\
 SI--fixed BCV &5.1&4.8&5.7&5.6&5.8 &4.9&4.9&5.7&5.6&5.7\\
 SI--adaptive CV &5.5&5.0&5.4&4.9&5.0 &5.2&4.3&4.8&4.2&4.7\\
 SI--adaptive BCV&5.7&5.0&5.5&5.8&6.5 &5.4&4.9&5.2&5.7&6.5\\
 PoSI1           &3.9&4.7&2.3&2.2&2.4 &4.0&4.9&2.6&2.6&2.7\\
 $t_{\mathrm{svy}}$ &5.2&5.7&6.9&11.6&37.4 &5.5&5.8&6.2&7.6&14.1\\
 $n_{\mathrm{sel}}$ fixed CV  &2999&3000&2994&2990&2977 &2999&3000&2994&2990&2977\\
 $n_{\mathrm{sel}}$ fixed BCV &2373&2327&2274&2107&1802 &2373&2327&2274&2107&1802\\
 $n_{\mathrm{sel}}$ adaptive CV  &2994&2977&2973&2951&2911 &2994&2977&2973&2951&2911\\
 $n_{\mathrm{sel}}$ adaptive BCV &2325&2278&2183&2050&1780 &2325&2278&2183&2050&1780\\
 \midrule
 \multicolumn{11}{c}{$n=800$}\\
 \midrule
 SI--fixed CV  &5.4&5.4&5.4&5.1&5.2 &5.3&4.9&4.9&4.6&4.9\\
 SI--fixed BCV &5.9&5.5&4.7&5.1&5.5 &5.3&5.0&4.1&4.7&5.1\\
 SI--adaptive CV &5.6&5.4&5.7&4.9&5.9 &5.2&5.2&5.0&4.5&5.4\\
 SI--adaptive BCV&5.6&5.5&4.7&5.0&5.5 &5.1&5.0&4.2&4.6&5.0\\
 PoSI1           &3.4&4.5&2.6&2.9&2.7 &3.3&4.7&2.6&2.8&2.9\\
 $t_{\mathrm{svy}}$ &5.5&5.8&7.4&15.4&58.9 &5.5&5.5&5.0&6.5&13.6\\
 $n_{\mathrm{sel}}$ fixed CV  &3000&3000&3000&3000&3000 &3000&3000&3000&3000&3000\\
 $n_{\mathrm{sel}}$ fixed BCV &2986&2984&2977&2972&2929 &2986&2984&2977&2972&2929\\
 $n_{\mathrm{sel}}$ adaptive CV  &3000&3000&3000&3000&2999 &3000&3000&3000&3000&2999\\
 $n_{\mathrm{sel}}$ adaptive BCV &2975&2976&2971&2960&2913 &2975&2976&2971&2960&2913\\
 \bottomrule
 \end{tabular}
 \end{center}
 \footnotesize{\textit{Notes}: $3{,}000$ Monte Carlo replications.  For the
 fixed rows, the CV and BCV penalties are the medians from 49 independent
 auxiliary samples and are held fixed during inference.  For the adaptive rows,
 $\hat\lambda$ is chosen on the same inference sample; adaptive CV and BCV use
 the same 10-fold split within a replication.  These adaptive rows are not
 covered by Proposition~\ref{prop: SI validity}.  BCV uses 500
 multiplier-bootstrap draws.  PoSI1 uses $B=999$ stratified bootstrap draws.
 $n_{\mathrm{sel}}$ is the number of replications in which the target regressor $\tilde{x}_{i1}$ is selected; SI rejection frequencies are conditional on target selection.}
\end{table}

\begin{table}[htbp]
	\caption{Empirical rejection frequencies (\%) of the tests for
		$H_0:\theta_{(2)}=1$ and $H_0:\mathrm{AME}_2=0.2449$
		at the $5\%$ level. Exogenous stratification, bounded-continuous design.}
	\label{tab:LevelES}
	\begin{center}
		\footnotesize
		\begin{tabular}{lcccccccccc}
			\toprule
			& \multicolumn{5}{c}{$H_0:\theta_{(2)}=1$}
			& \multicolumn{5}{c}{$H_0:\mathrm{AME}_2=0.2449$}\\
			\cmidrule(lr){2-6}\cmidrule(lr){7-11}
			& \multicolumn{5}{c}{$p/n$}
			& \multicolumn{5}{c}{$p/n$}\\
			\cmidrule(lr){2-6}\cmidrule(lr){7-11}
			& .01 & .025 & .05 & .10 & .25
			& .01 & .025 & .05 & .10 & .25\\
			\midrule
			\multicolumn{11}{c}{$n=200$}\\
			\midrule
			SI--fixed CV
			&5.6&6.3&6.0&5.3&4.7
			&5.4&5.9&5.6&4.9&4.3\\
			SI--fixed BCV
			&5.9&4.8&6.0&6.8&5.5
			&5.3&4.7&6.0&6.5&5.3\\
			SI--adaptive CV
			&5.3&6.2&7.2&5.5&7.0
			&5.4&6.1&6.8&5.1&6.3\\
			SI--adaptive BCV
			&6.6&4.9&7.9&6.4&6.3
			&6.4&4.9&7.8&6.2&6.3\\
			PoSI1
			&3.3&2.4&2.5&2.3&1.2
			&3.8&3.0&3.2&2.4&1.4\\
			$t_{\mathrm{svy}}$
			&5.4&5.5&6.9&7.6&23.2
			&6.0&6.0&7.5&7.0&13.9\\
			$n_{\mathrm{sel}}$ fixed CV
			&2981&2970&2882&2846&2759
			&2981&2970&2882&2846&2759\\
			$n_{\mathrm{sel}}$ fixed BCV
			&1313&1110&923&825&600
			&1313&1110&923&825&600\\
			$n_{\mathrm{sel}}$ adaptive CV
			&2904&2806&2726&2588&2394
			&2904&2806&2726&2588&2394\\
			$n_{\mathrm{sel}}$ adaptive BCV
			&1242&1108&963&801&646
			&1242&1108&963&801&646\\
			\midrule
			\multicolumn{11}{c}{$n=400$}\\
			\midrule
			SI--fixed CV
			&6.0&5.2&4.4&5.9&5.7
			&5.7&4.8&4.0&5.4&5.3\\
			SI--fixed BCV
			&6.4&5.3&5.2&5.7&6.0
			&6.0&5.3&4.8&5.4&5.9\\
			SI--adaptive CV
			&6.0&5.2&5.3&5.7&5.1
			&5.8&4.6&4.7&5.0&4.9\\
			SI--adaptive BCV
			&6.6&5.2&5.1&4.8&6.3
			&6.2&5.0&4.9&4.8&6.4\\
			PoSI1
			&3.6&2.2&1.8&2.2&1.3
			&3.8&2.5&2.0&2.2&1.6\\
			$t_{\mathrm{svy}}$
			&5.9&5.2&6.5&10.4&35.8
			&6.0&5.4&6.2&7.3&14.1\\
			$n_{\mathrm{sel}}$ fixed CV
			&2999&2999&2997&2993&2972
			&2999&2999&2997&2993&2972\\
			$n_{\mathrm{sel}}$ fixed BCV
			&2544&2450&2295&2244&1876
			&2544&2450&2295&2244&1876\\
			$n_{\mathrm{sel}}$ adaptive CV
			&2996&2987&2970&2947&2910
			&2996&2987&2970&2947&2910\\
			$n_{\mathrm{sel}}$ adaptive BCV
			&2443&2388&2269&2124&1857
			&2443&2388&2269&2124&1857\\
			\midrule
			\multicolumn{11}{c}{$n=800$}\\
			\midrule
			SI--fixed CV
			&5.1&4.8&5.7&6.2&5.6
			&4.8&4.5&5.3&5.7&5.3\\
			SI--fixed BCV
			&5.4&5.2&5.3&5.6&5.9
			&4.7&4.8&5.0&5.3&5.3\\
			SI--adaptive CV
			&5.0&5.7&5.3&5.3&6.0
			&4.7&5.2&5.0&4.7&5.6\\
			SI--adaptive BCV
			&5.3&5.4&5.3&5.6&5.8
			&4.7&4.9&5.0&5.2&5.3\\
			PoSI1
			&2.6&3.0&2.5&2.9&2.0
			&2.7&3.1&2.6&2.9&2.2\\
			$t_{\mathrm{svy}}$
			&5.0&5.6&7.5&14.5&54.3
			&5.0&5.6&5.4&6.8&13.3\\
			$n_{\mathrm{sel}}$ fixed CV
			&3000&3000&3000&3000&3000
			&3000&3000&3000&3000&3000\\
			$n_{\mathrm{sel}}$ fixed BCV
			&2990&2985&2980&2964&2944
			&2990&2985&2980&2964&2944\\
			$n_{\mathrm{sel}}$ adaptive CV
			&3000&3000&2999&3000&2998
			&3000&3000&2999&3000&2998\\
			$n_{\mathrm{sel}}$ adaptive BCV
			&2985&2983&2975&2956&2917
			&2985&2983&2975&2956&2917\\
			\bottomrule
		\end{tabular}
	\end{center}
	\footnotesize{\textit{Notes}: see
		Table~\ref{tab:Level0}.}
\end{table}

\section{Empirical application}\label{sec: app}

We apply the proposed methods to the corrected 2020 Canadian Internet
Use Survey (CIUS) Public Use Microdata File (PUMF) to examine the
demographic and socioeconomic correlates of accessing an online
account for a government program or service.\footnote{
	Statistics Canada released the 2020 CIUS PUMF on August 4, 2022
	(catalogue no.~56M0003X) and corrected it on December 12, 2024,
	restoring omitted variables and incomplete bootstrap replicate
	weights. The corrections permit a more granular covariate
	specification and design-based inference using the complete
	replicate weights. See
	\url{https://www150.statcan.gc.ca/n1/pub/56m0003x/56m0003x2020001-eng.htm}.}

\paragraph*{Data.}

The outcome is the response to ``During the past 12 months, did you
access an account for a government program or service on the Internet?''
We code reported government-account access as one, reported
nonaccess as zero and assign zero to the 2{,}115
\emph{Valid skip} cases, corresponding to respondents who reported
no personal Internet use in the preceding three months (see footnote~\ref{fn:validskip} for a sensitivity analysis
excluding these observations).
We exclude 378 \emph{Not stated} outcome responses, consistent
with the nonrejection of the Hausman--McFadden-type independence of irrelevant alternatives (IIA) diagnostic reported below, leaving
$n=17{,}031$ observations, of which 5{,}939 report access.

The specification contains 46 binary regressors for age,
urban or rural residence, gender, Indigenous identity,
knowledge of official languages, employment, education,
visible-minority status, household composition,
family-income quintile, immigration status, disability,
self-rated health, and province. Available
\emph{Not stated} covariate indicators are retained.
Reference categories are age 45--54, large urban centre,
male, non-Indigenous, neither official language, not employed,
some post-secondary education, not a visible minority,
a family with children under 18, income quintile~2,
landed immigrant, no disability, good health, and Alberta.
The rural/small-town indicator also includes the separately
classified Prince Edward Island cases; a provincial PEI
indicator enters independently.

The CIUS uses stratified probability sampling, but confidential
strata identifiers are unavailable in the PUMF. We use the
person-level survey weight and all 1{,}000 official bootstrap
replicate weights for design-based covariance estimation.\footnote{See the CIUS 2020 Microdata User Guide, Section~10.}

To assess the sensitivity of our results to excluding the 378
\emph{Not stated} outcome responses, we implement a
Hausman--McFadden-type IIA diagnostic using a survey-design
analogue of \texttt{Stata}'s \texttt{suest} procedure. The test compares
the \emph{Yes}-versus-\emph{No} coefficients from the
three-category multinomial logit and the restricted binary logit,
accounting for their joint covariance through first-order
linearization using the CIUS bootstrap replicate weights.
Among respondents with a \emph{Not stated} outcome, none
identified as Indigenous or reported poor health, resulting
in quasi-separation and a nearly singular Hessian.
We therefore exclude the two Indigenous-identity and five
self-rated-health indicators in the primary test. The resulting
39-regressor specification does not reject coefficient equality,
with a $p$-value of 0.488. A supplementary calculation using all
46 regressors and a Moore--Penrose generalized inverse of the
multinomial Hessian in constructing the sandwich covariance
estimator also yields nonrejection, with a $p$-value of 0.739.

\paragraph*{Estimation and inference.}
We fit the survey-weighted logit model~\eqref{eq: logitlasso}
with \texttt{glmnet}, standardizing regressors by
survey-weighted means and standard deviations.
The penalty parameter is selected using the survey-weighted
adaptation of the BCV procedure described in
Section~\ref{sec: simul}, with ten-fold survey-weighted
deviance cross-validation and 1{,}000 Gaussian
multiplier draws.
The procedure yields
$\hat\lambda_{\mathrm{BCV}}=0.023238$
and selects eight regressors.

\citet{ChetverikovSorensen2025} establish their
theoretical guarantees without survey weights, and we
do not extend them to the survey-weighted adaptation.
Moreover, because $\hat\lambda_{\mathrm{BCV}}$
is selected from the same sample, the empirical
implementation is not formally covered by
Proposition~\ref{prop: SI validity}, which
assumes a deterministic tuning parameter.
The adaptive BCV results in Section~\ref{sec: simul}
provide a finite-sample assessment of this
implementation: their empirical sizes are close to
the nominal level for $n\geq400$, well below the
sample size used here.

SI inference conditions on the selected model and coefficient signs;
under the conditions of Corollary~\ref{cor: SI model validity}, the same
intervals are also valid conditional on the selected model alone.
We estimate the pseudo-score covariance using the CIUS bootstrap
replicate weights. For population AMEs, the covariance estimator
additionally accounts for variation in the weighted average of
individual marginal effects and its covariance with the pseudo-score.
The $t_{\mathrm{svy}}$ columns report two-sided Wald p-values
based on full-model survey estimates and their replicate-weight
standard errors.

For comparison, we construct a computationally feasible PoSI1 model
family comprising all $2^8-1=255$ nonempty subsets of the eight
selected regressors and 1{,}000 randomly drawn nonempty models
from the full pool of 46 regressors. After removing duplicates,
the family contains 1{,}225 models, each including an intercept.
Coefficient and AME critical values are computed separately
from maxima over models containing each target regressor.
For AMEs, the computation accounts for both the first-order
estimation and empirical-averaging terms and their covariance.
Of the 1{,}000 CIUS bootstrap replicates, 998 yield finite
coefficient and AME estimates across all candidate models
and are used to compute the critical values.

Since exhaustive enumeration of all $2^{46}-1$ nonempty
models is computationally prohibitive, we use the
partly data-dependent model family described above
and report PoSI1 as an empirical benchmark rather
than invoke the simultaneous-coverage guarantee
for a prespecified model family.

\paragraph*{Results.}
Table~\ref{tab: coefs} presents full-model survey estimates,
$\ell_1$-PML estimates, and SI results for the candidate covariates.
Nine \emph{Not stated} covariate indicators are included in
estimation but omitted from the table for brevity;
none was selected by BCV. BCV selects Age 25--34, Age 35--44, Age 65+, French only,
High school or less, University degree, Visible minority, and Single.
BCV selects Age 25--34, Age 35--44, Age 65+, French,
High school or less, University degree, Visible minority,
and Single.\footnote{\label{fn:validskip}
	As a sensitivity check, we exclude the 2{,}115
	\emph{Valid skip} observations, reducing the sample
	to $n=14{,}916$. Reapplying BCV selects seven of the
	eight regressors selected in the main analysis;
	Single is no longer selected. The SI AMEs of the
	seven retained regressors preserve their signs,
	with six significant at the 5\% level.
	Visible minority remains insignificant
	(p-value 0.097). All seven variable-specific
	PoSI1 AME intervals exclude zero, whereas intervals
	based on a common critical value exclude zero for six.} Seven have two-sided SI p-values below 0.05 for both coefficients
and AMEs; the visible-minority AME has a $p$-value of 0.405.

The SI AMEs for Age 25--34, Age 35--44, and Age 65+ are 10.8,
8.9, and $-15.9$ percentage points, respectively. The
university-degree AME is 8.2 percentage points, whereas the
French-only, high-school-or-less, and single indicators have AMEs
of $-11.8$, $-9.6$, and $-8.2$ percentage points, respectively.
For indicators within a categorical block, the AME changes one dummy
while holding the others fixed, so it need not equal a contrast between
two feasible category levels. 
The SI estimates target selected-model pseudo-true
coefficients and their corresponding population AMEs.
A causal interpretation of these targets would require
additional identifying assumptions.

Table~\ref{tab:SI_PoSI} compares SI and variable-specific PoSI1
intervals. All eight PoSI1 AME intervals exclude zero, although
the lower endpoint for visible-minority status is just 0.00154;
its SI interval includes zero. Taking the maximum additionally
over all eight selected regressors yields common bootstrap
critical values of 3.505 for coefficients and 3.516 for AMEs.
Using these common critical values, seven AME intervals exclude
zero, whereas the visible-minority interval
$[-0.0107,\,0.0896]$ includes zero, matching the SI
significance classifications.

\section{Conclusion}\label{sec: conc}

This paper has provided two main results. First, we have extended the
$\ell_1$-PML SI method to GLMs with survey weights and established its
asymptotic validity for pseudo-true selected-model parameters under an
i.n.i.d. triangular-array framework when the number of candidate
regressors grows with the sample size. The validity result does not
require sure screening: selected models may omit active full-model
regressors, in which case the target is the corresponding restricted
population minimizer. Second, we have established selective inference
for nonlinear parameter functionals, including average marginal effects.
Together, these results provide a framework for valid post-selection
inference on both coefficients and economically interpretable nonlinear
effects in survey-weighted GLMs.

\begin{table}[H]
\caption{Logit estimates and average marginal effects for government online service access.}
\label{tab: coefs}\label{tab: AME}
\footnotesize
	\setlength{\tabcolsep}{2.4pt}
	\begin{center}	
	\renewcommand{\arraystretch}{1}
	\begin{tabular}{@{}lrrrrrrrrr@{}}
		\toprule
		& \multicolumn{5}{c}{Logit coefficients}
		& \multicolumn{4}{c}{Average marginal effects}\\
		\cmidrule(lr){2-6}\cmidrule(lr){7-10}
		& \multicolumn{3}{c}{Estimates}
		& \multicolumn{2}{c}{$p$-values}
		& \multicolumn{2}{c}{Estimates}
		& \multicolumn{2}{c}{$p$-values}\\
		\cmidrule(lr){2-4}\cmidrule(lr){5-6}
		\cmidrule(lr){7-8}\cmidrule(lr){9-10}
		Variable
		& GLM & $\ell_1$-PML & SI
		& $t_{\mathrm{svy}}$ & SI
		& GLM & SI
		& $t_{\mathrm{svy}}$ & SI\\
		\midrule
		Intercept
		& $-0.24$ & $-0.32$ & $-0.36$ & 0.64 & 0.002 & --- & --- & --- & ---\\
		Rural/small town
		& $-0.09$ & $0.00$ & $-$ & 0.17 & $-$ & $-0.018$ & $-$ & 0.164 & $-$\\
		\addlinespace[2pt]
		\multicolumn{10}{@{}l}{\textit{Age (ref.: 45--54)}}\\
		Age 15--24
		& $0.21$ & $0.00$ & $-$ & 0.06 & $-$ & $0.046$ & $-$ & 0.066 & $-$\\
		Age 25--34
		& $0.48$ & $0.24$ & $0.47$ & $<0.001$ & 0.007 & $0.106$ & $0.108$ & $<0.001$ & 0.008\\
		Age 35--44
		& $0.38$ & $0.15$ & $0.39$ & $<0.001$ & 0.005 & $0.083$ & $0.089$ & $<0.001$ & 0.005\\
		Age 55--64
		& $-0.10$ & $0.00$ & $-$ & 0.21 & $-$ & $-0.021$ & $-$ & 0.204 & $-$\\
		Age 65+
		& $-0.73$ & $-0.65$ & $-0.74$ & $<0.001$ & $<0.001$ & $-0.154$ & $-0.159$ & $<0.001$ & $<0.001$\\
		Female
		& $-0.04$ & $0.00$ & $-$ & 0.39 & $-$ & $-0.009$ & $-$ & 0.386 & $-$\\
		Indigenous identity
		& $0.06$ & $0.00$ & $-$ & 0.64 & $-$ & $0.014$ & $-$ & 0.643 & $-$\\
		English only
		& $0.08$ & $0.00$ & $-$ & 0.88 & $-$ & $0.017$ & $-$ & 0.874 & $-$\\
		French only
		& $-0.33$ & $-0.23$ & $-0.54$ & 0.50 & $<0.001$ & $-0.070$ & $-0.118$ & 0.490 & $<0.001$\\
		English and French
		& $0.20$ & $0.00$ & $-$ & 0.69 & $-$ & $0.042$ & $-$ & 0.692 & $-$\\
		Employed
		& $0.05$ & $0.00$ & $-$ & 0.44 & $-$ & $0.011$ & $-$ & 0.443 & $-$\\
		High school or less
		& $-0.55$ & $-0.32$ & $-0.43$ & $<0.001$ & $<0.001$ & $-0.119$ & $-0.096$ & $<0.001$ & $<0.001$\\
		University degree
		& $0.35$ & $0.29$ & $0.36$ & $<0.001$ & $<0.001$ & $0.077$ & $0.082$ & $<0.001$ & $<0.001$\\
		Visible minority
		& $0.12$ & $0.04$ & $0.18$ & 0.14 & 0.40 & $0.025$ & $0.040$ & 0.142 & 0.405\\
		Family w/o~18
		& $-0.05$ & $0.00$ & $-$ & 0.43 & $-$ & $-0.012$ & $-$ & 0.435 & $-$\\
		Single
		& $-0.46$ & $-0.12$ & $-0.37$ & $<0.001$ & $<0.001$ & $-0.095$ & $-0.082$ & $<0.001$ & $<0.001$\\
		Other household type
		& $-0.09$ & $0.00$ & $-$ & 0.53 & $-$ & $-0.020$ & $-$ & 0.527 & $-$\\
		\addlinespace[2pt]
		\multicolumn{10}{@{}l}{\textit{Income (ref.: Q2, \$44,120--\$75,321)}}\\
		Q1: \$44,119 or less
		& $-0.03$ & $0.00$ & $-$ & 0.73 & $-$ & $-0.006$ & $-$ & 0.734 & $-$\\
		Q3: \$75,322--\$109,431
		& $-0.07$ & $0.00$ & $-$ & 0.35 & $-$ & $-0.015$ & $-$ & 0.346 & $-$\\
		Q4: \$109,432--\$162,799
		& $-0.05$ & $0.00$ & $-$ & 0.50 & $-$ & $-0.011$ & $-$ & 0.499 & $-$\\
		Q5: \$162,800 or more
		& $-0.23$ & $0.00$ & $-$ & 0.008 & $-$ & $-0.048$ & $-$ & 0.007 & $-$\\
		Non-landed immigrant
		& $-0.01$ & $0.00$ & $-$ & 0.89 & $-$ & $-0.003$ & $-$ & 0.886 & $-$\\
		Disabled
		& $-0.03$ & $0.00$ & $-$ & 0.81 & $-$ & $-0.005$ & $-$ & 0.812 & $-$\\
		Health: excellent
		& $0.06$ & $0.00$ & $-$ & 0.47 & $-$ & $0.012$ & $-$ & 0.467 & $-$\\
		Health: very good
		& $0.06$ & $0.00$ & $-$ & 0.34 & $-$ & $0.013$ & $-$ & 0.341 & $-$\\
		Health: fair
		& $0.09$ & $0.00$ & $-$ & 0.42 & $-$ & $0.018$ & $-$ & 0.426 & $-$\\
		Health: poor
		& $0.23$ & $0.00$ & $-$ & 0.17 & $-$ & $0.051$ & $-$ & 0.173 & $-$\\
		\addlinespace[2pt]
		\multicolumn{10}{@{}l}{\textit{Province (ref.: Alberta)}}\\
		NL
		& $0.19$ & $0.00$ & $-$ & 0.10 & $-$ & $0.042$ & $-$ & 0.105 & $-$\\
		PEI
		& $0.07$ & $0.00$ & $-$ & 0.55 & $-$ & $0.016$ & $-$ & 0.549 & $-$\\
		NS
		& $-0.17$ & $0.00$ & $-$ & 0.13 & $-$ & $-0.036$ & $-$ & 0.122 & $-$\\
		NB
		& $-0.25$ & $0.00$ & $-$ & 0.03 & $-$ & $-0.054$ & $-$ & 0.023 & $-$\\
		QC
		& $-0.23$ & $0.00$ & $-$ & 0.04 & $-$ & $-0.050$ & $-$ & 0.037 & $-$\\
		ON
		& $-0.13$ & $0.00$ & $-$ & 0.14 & $-$ & $-0.028$ & $-$ & 0.135 & $-$\\
		MB
		& $-0.20$ & $0.00$ & $-$ & 0.10 & $-$ & $-0.041$ & $-$ & 0.099 & $-$\\
		SK
		& $0.02$ & $0.00$ & $-$ & 0.88 & $-$ & $0.004$ & $-$ & 0.880 & $-$\\
		BC
		& $0.09$ & $0.00$ & $-$ & 0.37 & $-$ & $0.020$ & $-$ & 0.370 & $-$\\
		\bottomrule
	\end{tabular}
	\end{center}
		\footnotesize
		\noindent\textit{Notes}: $n=17{,}031$.
		GLM and $t_{\mathrm{svy}}$ denote full-model survey estimates and t/Wald $p$-values.
		$\ell_1$-PML is the BCV-penalized estimator; SI reports
		selection-conditional one-step estimates and two-sided $p$-values.
		BCV selects eight of 46 regressors; nine ``not stated'' indicators
		enter the fitted model but are omitted from the display.
		AME inference incorporates estimation and empirical-averaging uncertainty.
		``$-$'' denotes nonselection; --- denotes an inapplicable intercept AME.
		The location indicator also includes PEI; province indicators enter separately.
		For categorical regressors, AMEs change one dummy at a time, holding others fixed.
\end{table}

\begin{table}[H]
	\small
	\caption{SI and PoSI1 confidence intervals for logit coefficients and average marginal effects, government online service access.}
	\label{tab:SI_PoSI}
	\setlength{\tabcolsep}{4pt}
	\begin{center}
		\begin{tabular}{@{}lrrrrr@{}}
			\toprule
			& \multicolumn{2}{c}{SI} & \multicolumn{3}{c}{PoSI1}\\
			\cmidrule(lr){2-3}\cmidrule(lr){4-6}
			Variable & 1-step est. & 95\% CI & Estimate & s.e. & 95\% CI\\
			\midrule
			\multicolumn{6}{l}{{Panel A: Logit coefficients}}\\
			\midrule
			Intercept & $-0.36$ & $[-0.45,\,-0.16]$ & --- & --- & ---\\
			Age 25--34 & $0.47$ & $[0.16,\,0.61]$ & $0.48$ & $0.07$ & $[0.28,\,0.67]$\\
			Age 35--44 & $0.39$ & $[0.15,\,0.51]$ & $0.39$ & $0.06$ & $[0.22,\,0.57]$\\
			Age 65+ & $-0.74$ & $[-0.98,\,-0.64]$ & $-0.76$ & $0.05$ & $[-0.92,\,-0.60]$\\
			French only & $-0.54$ & $[-0.67,\,-0.41]$ & $-0.59$ & $0.08$ & $[-0.83,\,-0.35]$\\
			High school or less & $-0.43$ & $[-0.55,\,-0.29]$ & $-0.44$ & $0.06$ & $[-0.61,\,-0.28]$\\
			University degree & $0.36$ & $[0.25,\,0.47]$ & $0.36$ & $0.06$ & $[0.21,\,0.51]$\\
			Visible minority & $0.18$ & $[-0.24,\,0.29]$ & $0.18$ & $0.06$ & $[0.01,\,0.35]$\\
			Single & $-0.37$ & $[-0.46,\,-0.27]$ & $-0.40$ & $0.05$ & $[-0.54,\,-0.25]$\\
			\midrule
			\multicolumn{6}{l}{{Panel B: Average marginal effects}}\\
			\midrule
			Age 25--34 & $0.108$ & $[0.036,\,0.141]$ & $0.107$ & $0.017$ & $[0.062,\,0.152]$\\
			Age 35--44 & $0.089$ & $[0.034,\,0.117]$ & $0.088$ & $0.014$ & $[0.048,\,0.128]$\\
			Age 65+ & $-0.159$ & $[-0.214,\,-0.138]$ & $-0.160$ & $0.011$ & $[-0.193,\,-0.128]$\\
			French only & $-0.118$ & $[-0.146,\,-0.089]$ & $-0.122$ & $0.015$ & $[-0.169,\,-0.075]$\\
			High school or less & $-0.096$ & $[-0.122,\,-0.065]$ & $-0.097$ & $0.013$ & $[-0.132,\,-0.061]$\\
			University degree & $0.082$ & $[0.057,\,0.108]$ & $0.081$ & $0.013$ & $[0.047,\,0.115]$\\
			Visible minority & $0.040$ & $[-0.056,\,0.067]$ & $0.039$ & $0.014$ & $[0.002,\,0.077]$\\
			Single & $-0.082$ & $[-0.104,\,-0.060]$ & $-0.084$ & $0.011$ & $[-0.113,\,-0.055]$\\
			\bottomrule
		\end{tabular}
	\end{center}
	\footnotesize{\textit{Notes}: The 95\% SI intervals invert the truncated-Gaussian pivot
		conditional on the selected model and coefficient signs. Population-AME
		inference includes the empirical weighted-average term and its covariance
		with the pseudo-score. PoSI1 estimates use the unpenalized
		survey-weighted selected-model fit and target-specific bootstrap
		critical values over the 1{,}225-model empirical candidate family.
		In row order, the coefficient critical values are
		$(2.658,2.810,2.964,3.087,2.626,2.680,2.645,2.783)$;
		the AME-specific values are
		$(2.679,2.812,2.931,3.103,2.652,2.713,2.655,2.761)$.
		The bootstrap maxima use 998 common finite replicates.
		The observed model family partly depends on selection, so these
		PoSI1 intervals are empirical benchmarks rather than claims of
		formal simultaneous coverage. The intercept is excluded from PoSI1.}
\end{table}
\bibliographystyle{myagsm}
\bibliography{SI}

\newpage
\appendix
\section{Proofs}\label{App: proofs}

\subsection{Proof of Lemma~\ref{lem:sample-matrix-consistency}}
\label{proof:sample-matrix-consistency}
Let $k\equiv d+p$ and define $A_{i,M}\equiv\left(w_i\ddot g(y_i,x_{iM}'\theta_M^*)\right)^{1/2}x_i
\in\mathbb R^k.$
Convexity of $g(y,\cdot)$ implies that $\ddot g(y,t)\geq0$. The upper bound $C_u$ on the survey weights (Assumption~\ref{A: AsyValid}\ref{AsyValid sampling}) and the bound
$|\ddot g(y,x_{iM}'\theta_M^*)|\leq C_g$ (Assumption~\ref{A: AsyValid}\ref{SIlocal derivatives}) together give, for every $u\in\mathbb R^k$ with $\|u\|_2=1$,
$|u'A_{i,M}|
\leq (C_uC_g)^{1/2}|u'x_i|$ a.s.
Consequently,
$\|u'A_{i,M}\|_{\psi_2}
\leq (C_uC_g)^{1/2}\|u'x_i\|_{\psi_2}
\leq K$,
where $K$ is independent of $i$, $M$, and $n$, using
	$\|x_i\|_{\psi_2}\leq C_u$ from
	Assumption~\ref{A: AsyValid}\ref{AsyValid sampling}. Thus, the
	independent vectors $A_{i,M}$ are uniformly sub-Gaussian with common
	norm bound $K$. Recall that
	\begin{equation*}
		\hat H(\theta_M^*)
		=
		n^{-1}\sum_{i=1}^n A_{i,M}A_{i,M}',
		\qquad
		H(\theta_M^*)
		=
		n^{-1}\sum_{i=1}^n\E[A_{i,M}A_{i,M}'],
	\end{equation*}
	and let
	\begin{equation*}
		B_M
		\equiv
		\hat H(\theta_M^*)-H(\theta_M^*)
		=
		n^{-1}\sum_{i=1}^n
		\left(
		A_{i,M}A_{i,M}'-\E[A_{i,M}A_{i,M}']
		\right).
	\end{equation*}
	
	\paragraph*{Step 1: Covering-net reduction.}
	Let $\mathcal N$ be a $1/4$-net of the unit sphere $\mathbb S^{k-1}$.
	By Corollary~4.2.13 of \cite{Vershynin(2018)}, $\mathcal N$ may be
	chosen with cardinality $|\mathcal N|\leq9^k$. Since $B_M$ is
	symmetric, Lemma~5.4 of \cite{Vershynin(2010)} gives
	\begin{equation}\label{eq:net-Hessian-bound}
		\|B_M\|_2
		\leq
		2\max_{u\in\mathcal N}|u'B_Mu|.
	\end{equation}
	For fixed $u\in\mathcal N$, set $Z_{i,u}=u'A_{i,M}$, so that
	\begin{equation*}
		u'B_Mu
		=
		n^{-1}\sum_{i=1}^n
		\left(Z_{i,u}^2-\E[Z_{i,u}^2]\right).
	\end{equation*}
	
	\paragraph*{Step 2: Sub-exponential norm of the summands.}
	The summands $\{Z_{i,u}^2-\E[Z_{i,u}^2]\}_{i=1}^n$ are independent
	and centered by construction. Since $Z_{i,u}$ is sub-Gaussian with
	$\|Z_{i,u}\|_{\psi_2}\leq K$, Lemma~5.14 of \cite{Vershynin(2010)}
	(equivalently, Lemma~2.7.6 of \cite{Vershynin(2018)}) gives
	\begin{equation*}
		\|Z_{i,u}^2\|_{\psi_1}
		\leq
		2\|Z_{i,u}\|_{\psi_2}^2
		\leq 2K^2,
	\end{equation*}
	and hence, by the triangle inequality for the sub-exponential norm,
	\begin{equation}\label{eq:subexp-bound}
		\|Z_{i,u}^2-\E[Z_{i,u}^2]\|_{\psi_1}
		\leq
		2\|Z_{i,u}^2\|_{\psi_1}
		\leq
		4K^2
		\equiv CK^2.
	\end{equation}
	
	\paragraph*{Step 3: Bernstein's inequality.}
	By \eqref{eq:subexp-bound}, the summands are independent, centered,
	and uniformly sub-exponential with norm bound $CK^2$. Bernstein's
	inequality (Corollary~5.17 of \cite{Vershynin(2010)}, or
	Corollary~2.8.3 of \cite{Vershynin(2018)}) therefore gives an
	absolute constant $c_1>0$, not depending on $u$, $i$, $M$, or $n$,
	such that for every $v>0$
	\begin{equation}\label{eq:Bernstein-Hessian}
		P\!\left[|u'B_Mu|>v\right]
		\leq
		2\exp\left\{
		-c_1n\min\left(\frac{v^2}{K^4},\frac{v}{K^2}\right)
		\right\}.
	\end{equation}
	
	\paragraph*{Step 4: Union bound.}
	Combining \eqref{eq:net-Hessian-bound} and
	\eqref{eq:Bernstein-Hessian} with a union bound over the
	$|\mathcal N|\leq9^k$ points of the net gives constants
	$c_0,C_0>0$, depending only on $K$, such that for every $t\geq0$,
	\begin{equation}\label{eq:Hessian-high-probability}
		P\left[
		\|B_M\|_2
		>
		C_0\left(
		\sqrt{\frac{k+t}{n}}+\frac{k+t}{n}
		\right)
		\right]
		\leq2\exp(-c_0t).
	\end{equation}
	Since $d$ is fixed, $k=d+p$, and $p/n\to0$,
	\begin{equation*}
		\|\hat H(\theta_M^*)-H(\theta_M^*)\|_2
		=
		O_p\!\left(\sqrt{\frac pn}\right).
	\end{equation*}
	
	\paragraph*{Design bound.}
	Steps~1--4 applied with $A_{i,M}$ replaced by $x_i$ (for which
	$\|x_i\|_{\psi_2}\leq C_u$ by
	Assumption~\ref{A: AsyValid}\ref{AsyValid sampling} plays the role
	of $K$) give, by the identical argument,
	\begin{equation*}
		\left\|
		n^{-1}X'X-\E[n^{-1}X'X]
		\right\|_2
		=
		O_p\!\left(\sqrt{\frac pn}\right).
	\end{equation*}
	Moreover, since the sub-Gaussian norm controls the second moment up
	to an absolute constant (Lemma~5.5 of \cite{Vershynin(2010)}),
	\begin{align*}
		\left\|\E[n^{-1}X'X]\right\|_2
		=
		\sup_{\|u\|_2=1}
		n^{-1}\sum_{i=1}^n\E[(u'x_i)^2]
		\leq
		C\sup_{\substack{1\leq i\leq n\\\|u\|_2=1}}
		\|u'x_i\|_{\psi_2}^2
		\leq C.
	\end{align*}
	It follows that
	\begin{equation*}
		\|n^{-1}X'X\|_2=O_p(1),
	\end{equation*}
	which proves the lemma.
$\square$

\subsection{A rate for the restricted estimator}

\begin{lemma}[Restricted model rate]\label{lem:restricted-rate}
	Suppose that Assumption~\ref{A: AsyValid} holds, and let
	$M=M_n\in\mathcal M_n(m_n)$ be deterministic. If $p/n\to0$ and
	$m_n(n^{-1/2}+\lambda)\to0$, then
	\begin{align}
		\|\hat\theta_M-\theta_M^*\|_2
		&=
		O_p\!\left(
		\sqrt{\frac{|M|}{n}}+\sqrt{|M|-d}\,\lambda
		\right),
		\label{eq:restricted-L2}\\
		\|\hat\theta_M-\theta_M^*\|_1
		&=
		O_p\!\left(
		\frac{|M|}{\sqrt n}+\sqrt{|M|(|M|-d)}\,\lambda
		\right).
		\label{eq:restricted-L1}
	\end{align}
	Under the rate conditions of Proposition~\ref{prop: SI validity}, these bounds
	reduce to
	\begin{equation}\label{eq:restricted-simplified}
		\|\hat\theta_M-\theta_M^*\|_2
		=O_p(\sqrt{m_n}\lambda),
		\qquad
		\|\hat\theta_M-\theta_M^*\|_1
		=O_p(m_n\lambda).
	\end{equation}
\end{lemma}

\begin{proof}
By the population first-order condition, $\E[S_M(\theta_M^*)]=0$.
Write
\begin{equation*}
S_M(\theta_M^*)=n^{-1}\sum_{i=1}^n\psi_{i,M},\qquad
\psi_{i,M}\equiv-w_ix_{iM}\dot g(y_i,x_{iM}'\theta_M^*),\qquad
\mu_{i,M}\equiv\E[\psi_{i,M}].
\end{equation*}
The individual means need not vanish,
but $\sum_i\mu_{i,M}=0$. By independence,
$\E[\psi_{i,M}'\psi_{k,M}]=\mu_{i,M}'\mu_{k,M}$ for $i\neq k$, and
\begin{equation*}
\sum_{i\neq k}\mu_{i,M}'\mu_{k,M}
=\Bigl\|\sum_{i=1}^n\mu_{i,M}\Bigr\|_2^2-\sum_{i=1}^n\|\mu_{i,M}\|_2^2
=-\sum_{i=1}^n\|\mu_{i,M}\|_2^2\leq0.
\end{equation*}
Hence, by the uniform moment bounds, for some constant $\bar{C}$,
\begin{align*}
\E\|S_M(\theta_M^*)\|_2^2
&\leq n^{-2}\sum_{i=1}^n
\E\left[w_i^2\|x_{iM}\|_2^2\,\dot{g}(y_i,x_{iM}'\theta_M^*)^2\right]\leq \bar{C}\frac{|M|}{n}.
\end{align*}
Thus,
\begin{equation}\label{eq:restricted-score-rate}
\|S_M(\theta_M^*)\|_2=O_p\!\left(\sqrt{\frac{|M|}{n}}\right).
\end{equation}
Lemma~\ref{lem:sample-matrix-consistency} and
Assumption~\ref{A: AsyValid}\ref{SIlocal curvature}
imply that, wpa 1,
\begin{equation}\label{eq:restricted-Hessian-target}
\|\hat H_M(\theta_M^*)-H_M\|_2=o_p(1),
\qquad
\lambda_{\min}\left(\hat H_M(\theta_M^*)\right)
\geq c_H/2.
\end{equation}
Let $r_{n,M}
\equiv A_0\left(
\sqrt{\frac{|M|}{n}}
+\sqrt{|M|-d}\,\lambda\right),$
where $A_0>0$ is a constant. For some constant $C$, uniformly over $\|\delta\|_2\leq r_{n,M}$
\begin{align*}
\max_i|x_{iM}'\delta|
&\leq C_u\|\delta\|_1
\leq C_u\sqrt{|M|}\,r_{n,M}\leq CA_0m_n(n^{-1/2}+\lambda)
=o(1).
\end{align*}
Thus, this ball is contained in the neighborhood specified in
Assumption~\ref{A: AsyValid}\ref{SIlocal derivatives}.
On this ball, the local Lipschitz condition and
Lemma~\ref{lem:sample-matrix-consistency} imply
\begin{equation}
\sup_{\|\delta\|_2\leq r_{n,M}}
\|\hat H_M(\theta_M^*+\delta)-\hat H_M(\theta_M^*)\|_2
\leq
C\sqrt{|M|}\,r_{n,M}
\|n^{-1}X_M'X_M\|_2=o_p(1).\label{eq:restricted-Hessian-local}
\end{equation}
It follows from \eqref{eq:restricted-Hessian-target},
\eqref{eq:restricted-Hessian-local}, and
Weyl's inequality \citep[Lemma 2.1]{Eaton-Tyler(1991)}
that wpa 1
\begin{equation}\label{eq:restricted-local-curvature}
\inf_{\|\delta\|_2\leq r_{n,M}}
\lambda_{\min}\left(
\hat H_M(\theta_M^*+\delta)
\right)
\geq c_H/4.
\end{equation}
Write $\delta=(\delta_\alpha',\delta_\beta')'$, where
$\delta_\alpha\in\mathbb R^d$ and
$\delta_\beta\in\mathbb R^{|M|-d}$ correspond to the
unpenalized and penalized components, respectively. By the second-order Taylor expansion with integral remainder,
\begin{align*}
-L_M(\theta_M^*+\delta)+L_M(\theta_M^*)=
	-S_M(\theta_M^*)'\delta
	+\int_0^1(1-t)\,
	\delta'\hat H_M(\theta_M^*+t\delta)\delta\,dt.
\end{align*}
Since, by \eqref{eq:restricted-local-curvature}, uniformly over $\|\delta\|_2\leq r_{n,M}$,
$
\lambda_{\min}\left(
\hat H_M(\theta_M^*+\delta)
\right)\geq \frac{c_H}{4}$ wpa 1, for every $t\in[0,1]$, $
\delta'\hat H_M(\theta_M^*+t\delta)\delta
\geq \frac{c_H}{4}\|\delta\|_2^2$ wpa 1.
It follows that wpa 1
\begin{equation*}
\int_0^1(1-t)\,
	\delta'\hat H_M(\theta_M^*+t\delta)\delta\,dt
\geq
	\frac{c_H}{4}\|\delta\|_2^2
	\int_0^1(1-t)\,dt
=
	\frac{c_H}{8}\|\delta\|_2^2.
	\end{equation*}
Therefore, by the Cauchy--Schwarz inequality, wpa 1
\begin{align*}
-L_M(\theta_M^*+\delta)+L_M(\theta_M^*)\geq
	-S_M(\theta_M^*)'\delta
	+\frac{c_H}{8}\|\delta\|_2^2\geq
	-\|S_M(\theta_M^*)\|_2\|\delta\|_2
	+\frac{c_H}{8}\|\delta\|_2^2.
\end{align*}
For the penalty term, the reverse triangle inequality
and the Cauchy--Schwarz inequality yield
\begin{align*}
	\|\beta_M^*+\delta_\beta\|_1-\|\beta_M^*\|_1\geq-\|\delta_\beta\|_1
	\geq-\sqrt{|M|-d}\,\|\delta_\beta\|_2
	\geq-\sqrt{|M|-d}\,\|\delta\|_2.
\end{align*}
Consequently,
\begin{align*}
&-L_M(\theta_M^*+\delta)
+\lambda\|\beta_M^*+\delta_\beta\|_1
+L_M(\theta_M^*)
-\lambda\|\beta_M^*\|_1\\
&\qquad\geq-\left(\|S_M(\theta_M^*)\|_2
+\sqrt{|M|-d}\,\lambda
\right)\|\delta\|_2
+\frac{c_H}{8}\|\delta\|_2^2.
\end{align*}
By \eqref{eq:restricted-score-rate}, the right-hand
side is positive uniformly on the sphere
$\|\delta\|_2=r_{n,M}$ with probability arbitrarily
close to one when $A_0$ is sufficiently large.
By convexity, the minimizer is inside the sphere,
proving \eqref{eq:restricted-L2}.
Equation~\eqref{eq:restricted-L1} follows from
$\|v\|_1\leq\sqrt{|M|}\,\|v\|_2$.
Finally, since $d$ is fixed and
	$\lambda\asymp\sqrt{(\log 2p)/n}\geq cn^{-1/2}$,
	\eqref{eq:restricted-simplified} follows.
\end{proof}

\subsection{Proof of Lemma~\ref{lem:pseudo-score-covariance}}
\label{proof:pseudo-score-covariance}

Fix a deterministic sequence
$M=M_n\in\mathcal M_n(m_n)$. Define the infeasible group means and
covariance estimator
\begin{align*}
\bar s_{h,M}^*
&\equiv
n_h^{-1}\sum_{i\in\mathcal I_h}s_{i,M},\quad\bar I_M^*
\equiv
n^{-1}\sum_{h=1}^J\frac{n_h}{n_h-1}
\sum_{i\in\mathcal I_h}
(s_{i,M}-\bar s_{h,M}^*)
(s_{i,M}-\bar s_{h,M}^*)'.
\end{align*}
By \eqref{eq: common group score mean}, the centered vectors
$s_{i,M}-\mu_{h,M}$ have mean zero within group $h$, although their
distributions and covariance matrices may differ across $i$. Define
$I_{h,M}\equiv n_h^{-1}\sum_{i\in\mathcal I_h}\operatorname{Var}(s_{i,M}).$
Independence and common within-group means imply
\begin{equation*}
	\E\!\left[
	\frac1{n_h-1}\sum_{i\in\mathcal I_h}
	(s_{i,M}-\bar s_{h,M}^*)
	(s_{i,M}-\bar s_{h,M}^*)'
	\right]
	=I_{h,M}.
\end{equation*}
The vectors $s_{i,M}-\mu_{h,M}$ are uniformly sub-Gaussian by
Assumption~\ref{A: AsyValid}\ref{AsyValid sampling}--\ref{AsyValid convex}. Applying the covariance
concentration argument from the proof of
Lemma~\ref{lem:sample-matrix-consistency} to the i.n.i.d. vectors
$s_{i,M}-\mu_{h,M}$ in each group, and using
$\|\bar s_{h,M}^*-\mu_{h,M}\|_2^2=O_p((d+p)/n_h)$ to account for
centering at the sample mean, yields
\begin{align*}
	\left\|
	\frac1{n_h-1}\sum_{i\in\mathcal I_h}
	(s_{i,M}-\bar s_{h,M}^*)
	(s_{i,M}-\bar s_{h,M}^*)'
	-I_{h,M}
	\right\|_2
	=
	O_p\!\left(
	\sqrt{\frac{d+p}{n_h}}+\frac{d+p}{n_h}
	\right).
\end{align*}
Because $J$ is fixed, $\min_h n_h/n\geq c_s$, and
$I_M=\sum_h(n_h/n)I_{h,M}$,
\begin{equation}\label{eq:infeasible-I-rate}
	\|\bar I_M^*-I_M\|_2
	=
	O_p\!\left(\sqrt{\frac pn}\right).
\end{equation}
Assumption~\ref{A: AsyValid}\ref{SIlocal variance} then gives
\begin{equation}\label{eq:infeasible-I-bounded}
	\|\bar I_M^*\|_2=O_p(1).
\end{equation}
It remains to control the replacement of $\theta_M^*$ by
$\hat\theta_M$. Let
\begin{equation*}
	e_{i,M}\equiv\hat s_{i,M}-s_{i,M},
	\qquad
	\bar e_{h,M}
	\equiv
	n_h^{-1}\sum_{i\in\mathcal I_h}e_{i,M}.
\end{equation*}
By the mean value theorem, for some $\tilde t_{i,M}$ between
$x_{iM}'\theta_M^*$ and $x_{iM}'\hat\theta_M$, we have 
$e_{i,M}=-w_i x_i\ddot g(y_i,\tilde t_{i,M})x_{iM}'(\hat\theta_M-\theta_M^*).$
From Lemma~\ref{lem:restricted-rate}, $
	\max_i|x_{iM}'(\hat\theta_M-\theta_M^*)|
	\leq
	C_u\|\hat\theta_M-\theta_M^*\|_1
	=
	O_p\left(m_n(n^{-1/2}+\lambda)\right)
	=o_p(1).$
Thus, the intermediate predictors lie in the neighborhood in
Assumption~\ref{A: AsyValid}\ref{SIlocal derivatives} wpa 1. We may therefore write
\begin{equation}\label{eq:score-plug-in-scalar}
	e_{i,M}=a_{i,M}x_i,
	\qquad
	\max_i|a_{i,M}|
	=
	O_p\left(m_n(n^{-1/2}+\lambda)\right).
\end{equation}
Define
\begin{equation*}
	R_{I,M}
	\equiv
	\frac1n\sum_{h=1}^J\frac{n_h}{n_h-1}
	\sum_{i\in\mathcal I_h}
	(e_{i,M}-\bar e_{h,M})
	(e_{i,M}-\bar e_{h,M})'.
\end{equation*}
For each group, $\sum_{i\in\mathcal I_h}(e_{i,M}-\bar e_{h,M})(e_{i,M}-\bar e_{h,M})'
=\sum_{i\in\mathcal I_h}e_{i,M}e_{i,M}'-n_h\bar e_{h,M}\bar e_{h,M}'$, so (in matrix sense)
\begin{equation*}
	\sum_{i\in\mathcal I_h}
	(e_{i,M}-\bar e_{h,M})(e_{i,M}-\bar e_{h,M})'
	\leq 
	\sum_{i\in\mathcal I_h}e_{i,M}e_{i,M}'
	=
	\sum_{i\in\mathcal I_h}a_{i,M}^2x_ix_i'
	\leq 
	\{\max_i|a_{i,M}|\}^2\sum_{i\in\mathcal I_h}x_ix_i'.
\end{equation*}
Summing over $h$, using $J$ fixed and $\min_h n_h/n\geq c_s$ to bound
$n_h/(n_h-1)$, together with \eqref{eq:design-matrix-bound} and
\eqref{eq:score-plug-in-scalar}, gives
\begin{equation}\label{eq:score-error-covariance}
	\|R_{I,M}\|_2
	\leq
	C\{\max_i|a_{i,M}|\}^2\|n^{-1}X'X\|_2
	=
	O_p\!\left(m_n^2(n^{-1/2}+\lambda)^2\right).
\end{equation}
Since $\hat s_{i,M}-\bar s_{h,M}
=(s_{i,M}-\bar s_{h,M}^*)+(e_{i,M}-\bar e_{h,M})$,
expanding $\hat I_M^{\mathrm{grp}}-\bar I_M^*$ gives two cross terms and
the quadratic term $R_{I,M}$. By the Cauchy--Schwarz inequality for sums
of outer products,
$\|\sum_ia_ib_i'\|_2\leq\|\sum_ia_ia_i'\|_2^{1/2}\|\sum_ib_ib_i'\|_2^{1/2}$,
\begin{equation*}
	\|\hat I_M^{\mathrm{grp}}-\bar I_M^*\|_2
	\leq
	2\|\bar I_M^*\|_2^{1/2}\|R_{I,M}\|_2^{1/2}
	+\|R_{I,M}\|_2.
\end{equation*}
It follows from \eqref{eq:infeasible-I-bounded},
\eqref{eq:score-error-covariance}, and
$m_n(n^{-1/2}+\lambda)\to0$ that
\begin{equation*}
	\|\hat I_M^{\mathrm{grp}}-\bar I_M^*\|_2
	=
	O_p\!\left(m_n(n^{-1/2}+\lambda)\right).
\end{equation*}
Since $\lambda\geq c_\lambda n^{-1/2}$, the last bound is
$O_p(m_n\lambda)$. Combining it with
\eqref{eq:infeasible-I-rate} proves \eqref{eq:I-consistency}. 

For a population AME, let
$u_{i,M}=w_i\{\mathrm{ME}_{ij}(\theta_M^*)
-\rho_M(\theta_M^*)\}$ and
$\bar w_n=n^{-1}\sum_i\E[w_i]$.
By \eqref{eq: AME influence term} and the
additional group-mean condition, group centering
$(s_{i,M}',u_{i,M}/\bar w_n)'$ yields an
unbiased infeasible estimator of $\Omega_M$.
Since the marginal effects and weights are bounded,
the preceding concentration argument gives a spectral norm error of $O_p(\sqrt{p/n})$. Furthermore, bounded logit derivatives,
\eqref{eq:restricted-L1}, and
\eqref{eq: AME sample linearization} imply $
\max_i\left|
\hat\tau_{i,M}-u_{i,M}/\bar w_n
\right|=O_p(m_n\lambda).
$
Together with \eqref{eq:score-plug-in-scalar}
and \eqref{eq:design-matrix-bound}, this implies
that the grouped covariance matrix of the
augmented plug-in errors has spectral norm
$O_p(m_n^2\lambda^2)$.
Since the infeasible augmented covariance estimator
has bounded spectral norm in probability,
the covariance perturbation is $O_p(m_n\lambda)$
by Cauchy--Schwarz.
Combining the two bounds proves
\eqref{eq:Omega-consistency}.
\hfill$\square$

\subsection{Proof of Proposition~\ref{prop: SI validity}}
\label{subsec: Proof SI validity}

We prove the proposition in five steps: auxiliary
bounds, linear representations, Gaussian coupling, estimation of the
covariance and threshold, and transfer of the Gaussian pivot to the
selected sample.

\paragraph*{Step 1: Auxiliary bounds.}
On the event $E_{M,\sm}=\{\hat M=M,\hsm=\sm\},$
the estimator $\hat\theta_M$ equals the active block of the
full $\ell_1$-PML estimator. Thus, every fixed $M$ statistic below
coincides with its feasible counterpart on this event. By the triangle inequality,
\begin{equation}\label{eq:Hessian-triangle}
	\|\hat H(\hat\theta_M)-H(\theta_M^*)\|_2
	\leq
	\|\hat H(\hat\theta_M)-\hat H(\theta_M^*)\|_2
	+
	\|\hat H(\theta_M^*)-H(\theta_M^*)\|_2.
\end{equation}
The second term is $O_p(\sqrt{p/n})$ by \eqref{eq:Hessian-concentration}
in Lemma~\ref{lem:sample-matrix-consistency}. For the first term, the
local Lipschitz condition on $\ddot g$
(Assumption~\ref{A: AsyValid}\ref{SIlocal derivatives}) gives
\begin{align*}
	\|\hat H(\hat\theta_M)-\hat H(\theta_M^*)\|_2
	\leq
	C\max_i
	|x_{iM}'(\hat\theta_M-\theta_M^*)|
	\|n^{-1}X'X\|_2.
\end{align*}
Lemma~\ref{lem:restricted-rate} gives $
	\max_i|x_{iM}'(\hat\theta_M-\theta_M^*)|
	\leq
	C_u\|\hat\theta_M-\theta_M^*\|_1
	=
	O_p\left(m_n(n^{-1/2}+\lambda)\right),$
and \eqref{eq:design-matrix-bound} in
Lemma~\ref{lem:sample-matrix-consistency} gives $\|n^{-1}X'X\|_2=O_p(1)$.
Hence, $
	\|\hat H(\hat\theta_M)-\hat H(\theta_M^*)\|_2
	=
	O_p\left(m_n(n^{-1/2}+\lambda)\right)
	=
	O_p(m_n\lambda),$
where the last equality follows from
$\lambda=C\sqrt{(\log 2p)/n}$ in
Proposition~\ref{prop: SI validity}, which implies
$n^{-1/2}=O(\lambda)$ since $p\geq1$. Combining this with \eqref{eq:Hessian-triangle} and the bound on the
second term gives
\begin{equation}\label{eq:pseudo-Hessian-full-rate}
	\|\hat H(\hat\theta_M)-H(\theta_M^*)\|_2
	=
	O_p\!\left(\sqrt{\frac pn}+m_n\lambda\right)
	\equiv O_p(r_{H,n})=o_p(1).
\end{equation}
Assumption~\ref{A: AsyValid}\ref{SIlocal curvature}, Weyl's inequality,
and the resolvent identity (\eqref{eq: resolvent} below) now imply
\begin{align}
	\|\hat H_M(\hat\theta_M)^{-1}\|_2
	&=O_p(1),
	\label{eq:pseudo-H-inverse}\\
	\|\hat H_M(\hat\theta_M)^{-1}-H_M^{-1}\|_2
	&=O_p(r_{H,n}),
	\label{eq:pseudo-H-inverse-difference}\\
	\|\hat H_{-M,M}(\hat\theta_M)\hat H_M(\hat\theta_M)^{-1}
	-H_{-M,M}H_M^{-1}\|_2
	&=O_p(r_{H,n}).
	\label{eq:pseudo-H-cross-difference}
\end{align}
Indeed, since $\hat H_M(\hat\theta_M)$ and $H_M$ are the $(M,M)$
principal submatrices of $\hat H(\hat\theta_M)$ and $H(\theta_M^*)$,
the spectral norm of a principal submatrix is bounded by that of the
full matrix, so \eqref{eq:pseudo-Hessian-full-rate} gives
\begin{equation}\label{eq:pseudo-HM-diff}
	\|\hat H_M(\hat\theta_M)-H_M\|_2
	\leq
	\|\hat H(\hat\theta_M)-H(\theta_M^*)\|_2
	=O_p(r_{H,n})=o_p(1).
\end{equation}
By Weyl's inequality,
$|\lambda_{\min}\left(\hat H_M(\hat\theta_M)\right)-\lambda_{\min}(H_M)|
\leq\|\hat H_M(\hat\theta_M)-H_M\|_2$, so \eqref{eq:pseudo-HM-diff}
and Assumption~\ref{A: AsyValid}\ref{SIlocal curvature}
($\lambda_{\min}(H_M)\geq c_H$) give, wpa 1, $\lambda_{\min}\left(\hat H_M(\hat\theta_M)\right)\geq c_H/2$, hence
\begin{equation}\label{eq: twobounds}
	\|\hat H_M(\hat\theta_M)^{-1}\|_2=O_p(1),
	\qquad
	\|H_M^{-1}\|_2\leq1/c_H=O(1),
\end{equation}
which is \eqref{eq:pseudo-H-inverse}. For
\eqref{eq:pseudo-H-inverse-difference}, the resolvent identity
\begin{equation}\label{eq: resolvent}
	\hat H_M(\hat\theta_M)^{-1}-H_M^{-1}
	=
	\hat H_M(\hat\theta_M)^{-1}
	\left(H_M-\hat H_M(\hat\theta_M)\right)
	H_M^{-1}
\end{equation}
and the Cauchy--Schwarz inequality for the spectral norm give
\begin{equation*}
	\|\hat H_M(\hat\theta_M)^{-1}-H_M^{-1}\|_2
	\leq
	\|\hat H_M(\hat\theta_M)^{-1}\|_2
	\|\hat H_M(\hat\theta_M)-H_M\|_2
	\|H_M^{-1}\|_2
	=
	O_p(1)O_p(r_{H,n}) O(1)
	=
	O_p(r_{H,n}),
\end{equation*}
using \eqref{eq:pseudo-HM-diff} and the two bounds in \eqref{eq: twobounds}. For \eqref{eq:pseudo-H-cross-difference}, write
\begin{align*}
	&\hat H_{-M,M}(\hat\theta_M)\hat H_M(\hat\theta_M)^{-1}
	-H_{-M,M}H_M^{-1}\\
	&=
	\left(\hat H_{-M,M}(\hat\theta_M)-H_{-M,M}\right)
	\left(\hat H_M(\hat\theta_M)^{-1}-H_M^{-1}\right)
	+
	\left(\hat H_{-M,M}(\hat\theta_M)-H_{-M,M}\right)H_M^{-1}\\
	&\quad
	+H_{-M,M}\left(\hat H_M(\hat\theta_M)^{-1}-H_M^{-1}\right)
	\equiv\mathrm I+\mathrm{II}+\mathrm{III}.
\end{align*}
As in \eqref{eq:pseudo-HM-diff}, the submatrix bound gives
$\|\hat H_{-M,M}(\hat\theta_M)-H_{-M,M}\|_2
\leq\|\hat H(\hat\theta_M)-H(\theta_M^*)\|_2=O_p(r_{H,n})$, while
$\|H_{-M,M}\|_2\leq\|H(\theta_M^*)\|_2\leq C_H=O(1)$ by
Assumption~\ref{A: AsyValid}\ref{SIlocal curvature}. Hence,
\begin{equation*}
	\|\mathrm I\|_2=O_p(r_{H,n}^2)=o_p(r_{H,n}),
	\qquad
	\|\mathrm{II}\|_2=O_p(r_{H,n}),
	\qquad
	\|\mathrm{III}\|_2=O_p(r_{H,n}),
\end{equation*}
using \eqref{eq:pseudo-H-inverse-difference} for $\mathrm{II}$ and
$\mathrm{III}$, and $r_{H,n}=o(1)$ for $\mathrm I$. Summing the three
terms proves \eqref{eq:pseudo-H-cross-difference}.

	Also, \eqref{eq:restricted-score-rate} gives
	\begin{equation}\label{eq:pseudo-active-score-rate}
		\|n^{1/2}\,S_M(\theta_M^*)\|_2
		=O_p(\sqrt{|M|})
		=O_p(\sqrt{m_n}).
	\end{equation}
	
\paragraph*{Step 2: Linear representations.}
Let $
	\Delta_M\equiv\hat\theta_M-\theta_M^*.$
By the mean value expansion, for each $i$ there exists
$\tilde t_{i,M}$ between $x_{iM}'\theta_M^*$ and
$x_{iM}'\hat\theta_M$ such that
\begin{equation}
	\dot g(y_i,x_{iM}'\hat\theta_M)
	=
	\dot g(y_i,x_{iM}'\theta_M^*)
	+
	\ddot g(y_i,\tilde t_{i,M})\,x_{iM}'\Delta_M.
\end{equation}
Multiplying by $-w_i[x_i]_A$, where $[x_i]_A$ denotes the subvector
of $x_i$ indexed by $A$, and averaging over $i$ gives, for
$A\in\{M,-M\}$,
\begin{equation*}
	S_A(\hat\theta_M)
	=
	S_A(\theta_M^*)
	-
	\left(
	n^{-1}\sum_{i=1}^n
	w_i[x_i]_A
	\ddot g(y_i,\tilde t_{i,M})
	x_{iM}'
	\right)
	\Delta_M.
\end{equation*}
The matrix in the bracket involves the intermediate term
$\ddot g(y_i,\tilde t_{i,M})$, rather than the plug-in term
$\ddot g(y_i,x_{iM}'\hat\theta_M)$ defining
$\hat H_{A,M}(\hat\theta_M)$. Adding and subtracting
$\hat H_{A,M}(\hat\theta_M)\Delta_M$ therefore gives
\begin{align}
	S_M(\hat\theta_M)
	&=
	S_M(\theta_M^*)
	-\hat H_M(\hat\theta_M)\Delta_M
	+R_M,
	\label{eq:pseudo-score-expansion-active}\\
	S_{-M}(\hat\theta_M)
	&=
	S_{-M}(\theta_M^*)
	-\hat H_{-M,M}(\hat\theta_M)\Delta_M
	+R_{-M},
	\label{eq:pseudo-score-expansion-inactive}
\end{align}
where $
	R_A
	\equiv
	-n^{-1}\sum_{i=1}^n
	w_i[x_i]_A
	\left(
	\ddot g(y_i,\tilde t_{i,M})
	-
	\ddot g(y_i,x_{iM}'\hat\theta_M)
	\right)
	x_{iM}'\Delta_M, 
	A\in\{M,-M\}.$
We next bound $R_A$. By
Assumption~\ref{A: AsyValid}\ref{AsyValid sampling} and
\eqref{eq:restricted-simplified},
\begin{equation*}
	\max_i|x_{iM}'\Delta_M|
	\leq
	C_u\|\Delta_M\|_1
	=
	O_p(m_n\lambda)
	=
	o_p(1).
\end{equation*}
Hence, wpa 1,
$x_{iM}'\hat\theta_M$ and every intermediate point
$\tilde t_{i,M}$ lie in the neighborhood on which
Assumption~\ref{A: AsyValid}\ref{SIlocal derivatives} applies.
In consequence,
\begin{align*}
	&|\ddot g(y_i,\tilde t_{i,M})
	-\ddot g(y_i,x_{iM}'\hat\theta_M)|
\leq
	L_0|\tilde t_{i,M}-x_{iM}'\hat\theta_M|
	\leq
	L_0|x_{iM}'\Delta_M|
	\leq
	C\|\Delta_M\|_1.
\end{align*}
Thus, the remainder is quadratic in $\Delta_M$.
To make the rate explicit, define
\begin{equation*}
	D_M
	\equiv
	\operatorname{diag}\!\left[
	w_i\left(
	\ddot g(y_i,\tilde t_{i,M})
	-
	\ddot g(y_i,x_{iM}'\hat\theta_M)
	\right)
	\right]_{i=1}^n.
\end{equation*}
The stacked remainder can be rewritten as (up to a permutation of coordinates which does not affect the Euclidean norm) $(R_M',R_{-M}')'=-n^{-1}X'D_MX_M\Delta_M.$
Since $D_M$ is diagonal, boundedness of the weights and the preceding
Lipschitz bound imply
$\|D_M\|_2
	=
	\max_i|(D_M)_{ii}|
	\leq
	C\|\Delta_M\|_1
	=
	O_p(m_n\lambda).$
Therefore,
\begin{align*}
	\left\|
	n^{-1}X'D_MX_M\Delta_M
	\right\|_2
	&\leq
	\|D_M\|_2
	\|n^{-1}X'X\|_2^{1/2}
	\|n^{-1}X_M'X_M\|_2^{1/2}
	\|\Delta_M\|_2\\
	&=
	O_p(m_n\lambda)\, O_p(1)\, O_p(1)\, O_p(\sqrt{m_n}\lambda)\\
	&=
	O_p(m_n^{3/2}\lambda^2),
\end{align*}
where we use 
$\|n^{-1}X'X\|_2=O_p(1)$ (Lemma~\ref{lem:sample-matrix-consistency}) and
$\|n^{-1}X_M'X_M\|_2\leq\|n^{-1}X'X\|_2$.
Since
$\lambda=C\sqrt{(\log 2p)/n}$,
\begin{equation}\label{eq:pseudo-score-remainder}
	n^{1/2}
	\left\|
		(R_M',
		R_{-M}')'
	\right\|_2
	=
	O_p\!\left(
	\frac{m_n^{3/2}\log 2p}{n^{1/2}}
	\right)
	=
	o_p\!\left((\log 2p)^{-1/2}\right),
\end{equation}
because $
	\frac{m_n^{3/2}(\log 2p)^{3/2}}{n^{1/2}}
	\to0$
by the rate conditions of Proposition~\ref{prop: SI validity}.
Combining \eqref{eq: onestep},
\eqref{eq:pseudo-score-expansion-active}, and
\eqref{eq:pseudo-score-remainder} gives
\begin{equation*}
	n^{1/2}(\tilde\theta_M-\theta_M^*)
	=
	\hat H_M(\hat\theta_M)^{-1}
	n^{1/2}S_M(\theta_M^*)
	+
	o_p\!\left((\log2p)^{-1/2}\right).
\end{equation*}
By \eqref{eq:pseudo-H-inverse-difference} and
\eqref{eq:pseudo-active-score-rate},
\begin{align*}
	&\left\|
	\left(\hat H_M(\hat\theta_M)^{-1}-H_M^{-1}\right)
	n^{1/2}S_M(\theta_M^*)
	\right\|_2=
	O_p(r_{H,n}m_n^{1/2})
	=
	O_p\!\left(
	\sqrt{\frac{pm_n}{n}}
	+
	m_n^{3/2}\lambda
	\right)
	=
	o_p\!\left(\sqrt{\frac{1}{\log2p}}\right),
\end{align*}
where the last equality follows from
the rate conditions of Proposition~\ref{prop: SI validity}. Hence,
\begin{equation}\label{eq:pseudo-theta-linearization}
	n^{1/2}(\tilde\theta_M-\theta_M^*)
	=
	H_M^{-1}
	n^{1/2}S_M(\theta_M^*)
	+
	o_p\!\left((\log 2p)^{-1/2}\right).
\end{equation}
For the inactive score, let $
	\hat B_M
	\equiv
	\hat H_{-M,M}(\hat\theta_M)
	\hat H_M(\hat\theta_M)^{-1}$ and 
	$B_M\equiv
	H_{-M,M}H_M^{-1}.$
Substituting
\eqref{eq:pseudo-score-expansion-active} and
\eqref{eq:pseudo-score-expansion-inactive} into
\eqref{eq: effscore-M}, the terms involving $\Delta_M$ cancel exactly:
\begin{align*}
	\tilde S_{-M}(\hat\theta_M)
	&=
	S_{-M}(\theta_M^*)
	-\hat B_M S_M(\theta_M^*)
	+R_{-M}-\hat B_MR_M.
\end{align*}
Therefore,
\begin{align*}
	n^{1/2}\left(
	\tilde S_{-M}(\hat\theta_M)-\nu_{-M}
	\right)
	&=
	n^{1/2}\left(S_{-M}(\theta_M^*)-\nu_{-M}\right)
	-
	B_Mn^{1/2}S_M(\theta_M^*)\\
	&\qquad\quad
	-
	(\hat B_M-B_M)n^{1/2}S_M(\theta_M^*)
	+
	n^{1/2}(R_{-M}-\hat B_MR_M).
\end{align*}
By \eqref{eq:pseudo-H-cross-difference},
\eqref{eq:pseudo-active-score-rate}, and
\eqref{eq:pseudo-score-remainder},
\begin{equation*}
	\left\|(\hat B_M-B_M)n^{1/2}S_M(\theta_M^*)\right\|_2
	=
	O_p(r_{H,n}\sqrt{m_n})
	=
	o_p\!\left((\log 2p)^{-1/2}\right),
\end{equation*}
while $\|\hat B_M\|_2=O_p(1)$ implies $
n^{1/2}\left\|R_{-M}-\hat B_MR_M\right\|_2
=o_p\!\left((\log2p)^{-1/2}\right).$
It follows that
\begin{align}
	&n^{1/2}\left(
	\tilde S_{-M}(\hat\theta_M)-\nu_{-M}
	\right)=
	n^{1/2}\left(S_{-M}(\theta_M^*)-\nu_{-M}\right)
	-
	H_{-M,M}H_M^{-1}
	n^{1/2}S_M(\theta_M^*)
	+
	o_p\!\left((\log2p)^{-1/2}\right).
	\label{eq:pseudo-inactive-linearization}
\end{align}
For the general functional, Taylor expansion of the sample
functional and its locally Lipschitz Jacobian give
\begin{align*}
	n^{1/2}\left|
	\hat\rho_M(\hat\theta_M)-\hat\rho_M(\theta_M^*)
	-
	\nabla{\hat\rho}_M(\hat\theta_M)'
	(\hat\theta_M-\theta_M^*)
	\right|
	&\leq
	Cn^{1/2}\|\hat\theta_M-\theta_M^*\|_2^2\\
	&=
	O_p\!\left(
	\frac{m_n\log2p}{n^{1/2}}
	\right)
	=
	o_p\!\left((\log2p)^{-1/2}\right).
\end{align*}
Moreover, the empirical Jacobian condition in
Assumption~\ref{A: AsyValid}\ref{SIpivot rho}, its local Lipschitz
property, \eqref{eq:restricted-simplified}, and
\eqref{eq:pseudo-active-score-rate} imply
\begin{align*}
	&\left|
	\left(\nabla{\hat\rho}_M(\hat\theta_M)
	-\nabla\rho_M(\theta_M^*)\right)'
	H_M^{-1}
	n^{1/2}S_M(\theta_M^*)
	\right|=
	O_p(m_n\lambda)
	=
	o_p\!\left((\log2p)^{-1/2}\right).
\end{align*}
By Assumption~\ref{A: AsyValid}\ref{SIpivot rho}
and the influence-function representation
\eqref{eq: functional averaging expansion},
\begin{equation*}
	n^{1/2}\left(
	\hat\rho_M(\theta_M^*)-\rho_M(\theta_M^*)
	\right)
	=
	n^{-1/2}\sum_{i=1}^n\tau_{i,M}
	+
	o_p\left((\log 2p)^{-1/2}\right).
\end{equation*}
For functionals depending only on the parameter vector,
$\hat\rho_M=\rho_M$ and $\tau_{i,M}=0$.
Combining this representation with the preceding
Taylor expansion, the gradient bounds, and
\eqref{eq:pseudo-theta-linearization} yields
\begin{equation}\label{eq:pseudo-rho-linearization}
	n^{1/2}\left(
	\tilde\rho_M-\rho_M(\theta_M^*)
	\right)
	=
	\nabla\rho_M(\theta_M^*)'H_M^{-1}
	n^{1/2}S_M(\theta_M^*)
	+
	n^{-1/2}\sum_{i=1}^n\tau_{i,M}
	+
	o_p\left((\log 2p)^{-1/2}\right).
\end{equation}
Since
$\E[S_M(\theta_M^*)]=0$ and
$\E[S_{-M}(\theta_M^*)]=\nu_{-M}$,
\eqref{eq:pseudo-theta-linearization},
\eqref{eq:pseudo-inactive-linearization}, and
\eqref{eq:pseudo-rho-linearization} give
\begin{equation}\label{eq:pseudo-augmented-linearization}
	Z-\mu_M
	=
	Q_Mn^{-1/2}
	\sum_{i=1}^n\omega_{i,M}
	+
	r_{M,n},
	\qquad
	\|r_{M,n}\|_\infty
	=
	o_p\!\left((\log2p)^{-1/2}\right).
\end{equation}
	
\paragraph*{Step 3: Gaussian coupling.}
Set $W_M\equiv Q_Mn^{-1/2}\sum_{i=1}^n\omega_{i,M}.$ By Assumption~\ref{A: AsyValid}\ref{AsyValid sampling}--%
\ref{AsyValid convex}, the weights $w_i$ and
score derivatives $\dot g(y_i,x_{iM}'\theta_M^*)$
are uniformly bounded. Since
$s_{i,M}=-w_i\dot g(y_i,x_{iM}'\theta_M^*)x_i$
and $x_i$ is uniformly sub-Gaussian,
$\|s_{i,M}\|_{\psi_2}\leq C$.
Centering preserves this bound, so
$\|\xi_{i,M}\|_{\psi_2}\leq C$.
Assumption~\ref{A: AsyValid}\ref{SIpivot rho}
gives the same bound for $\tau_{i,M}$,
and hence $\|\omega_{i,M}\|_{\psi_2}\leq C$.

We first bound the population transformation matrix $Q_M$. By
\eqref{eq: Qrho population main},
Assumption~\ref{A: AsyValid}\ref{SIlocal curvature}, and
Assumption~\ref{A: AsyValid}\ref{SIpivot rho},
$\|H_M^{-1}\|_2
	\leq c_H^{-1}$
and 
	$\|H_{-M,M}H_M^{-1}\|_2
	\leq \frac{C_H}{c_H},$
because $
	\|H_{-M,M}\|_2
	\leq
	\|H(\theta_M^*)\|_2
	\leq C_H.$
Moreover,
\begin{equation*}
	\|\nabla\rho_M(\theta_M^*)'H_M^{-1}\|_2
	\leq
	\frac{\|\nabla\rho_M(\theta_M^*)\|_2}{c_H}
	\leq
	\frac{\sqrt{C_\rho}}{c_H}.
\end{equation*}
Since the remaining nonzero blocks of $Q_M$ are identity
matrices, it follows that
\begin{equation}\label{eq:Q-population-bound}
	\|Q_M\|_2=O(1).
\end{equation}
We next apply Theorem~\ref{thm: Yur} to $U_{i,M}\equiv n^{-1/2}Q_M\omega_{i,M}\in\mathbb R^{d+p+1}$ and 
$V_{i,M}\equiv\operatorname{Var}(U_{i,M}).$
By the uniform sub-Gaussian bounds established
above and \eqref{eq:Q-population-bound},
\begin{equation*}
	\|U_{i,M}\|_{\psi_2}
	\leq Cn^{-1/2},
	\qquad
	\max_j(V_{i,M})_{jj}\leq Cn^{-1},
	\qquad
	\operatorname{tr}(V_{i,M})\leq Cp/n.
\end{equation*}
The sub-Gaussian moment bounds imply
$\E\left[\|U_{i,M}\|_2^4\right]
	\leq Cp^2/n^2$ and $
	\E\left[\|U_{i,M}\|_\infty^2\right]
	\leq C\frac{\log2p}{n}.$
Hence, by Cauchy--Schwarz,
\begin{align*}
	\sum_{i=1}^n
	\E\left[
	\|U_{i,M}\|_2^2\|U_{i,M}\|_\infty
	\right]
	&\leq
	\sum_{i=1}^n
	\left(\E\left[\|U_{i,M}\|_2^4\right]\right)^{1/2}
	\left(\E\left[\|U_{i,M}\|_\infty^2\right]\right)^{1/2}=
	O\left(
	\frac{p\sqrt{\log2p}}{\sqrt n}
	\right).
\end{align*}
The Gaussian inequality in Lemma SA.3 of 
\cite{Cattaneo2024} gives
\begin{equation*}
\sum_{i=1}^n
	\E\left[
	\|V_{i,M}^{1/2}Z_i\|_2^2
	\|V_{i,M}^{1/2}Z_i\|_\infty
	\right]\leq
	C\sum_{i=1}^n
	\sqrt{\frac{\log2p}{n}}\,
	\operatorname{tr}(V_{i,M})
	=
	O\left(
	\frac{p\sqrt{\log2p}}{\sqrt n}
	\right).
\end{equation*}
Consequently,
\begin{equation}\label{eq:Yurinskii-beta}
	\beta_{\infty,2}
	=
	O\left(
	\frac{p\sqrt{\log2p}}{\sqrt n}
	\right).
\end{equation}
Since
$\phi_\infty(d+p+1)=O(\sqrt{\log2p})$
and $
	\sum_{i=1}^nV_{i,M}
	=Q_M\Omega_MQ_M'
	=\Sigma_M,$
Theorem~\ref{thm: Yur} implies that,
for any sequence $R_n\to\infty$,
there exists a coupling with
$T_M\sim N(0,\Sigma_M)$ such that
\begin{equation}\label{eq:pseudo-Yurinskii-rate}
	\|W_M-T_M\|_\infty
	=
	O_p\left(
	\frac{p^{1/3}\sqrt{\log2p}}{n^{1/6}}
	R_n
	\right).
\end{equation}
By the rate conditions of Proposition~\ref{prop: SI validity}, 
$\frac{p^{1/3}\log2p}{n^{1/6}}\to0.$
Hence, $R_n\to\infty$ can be chosen sufficiently slowly that $
\frac{p^{1/3}\sqrt{\log2p}}{n^{1/6}}
	R_n
	=
	o\!\left(
	\left(\log2p\right)^{-1/2}
	\right).$
Therefore, $\|W_M-T_M\|_\infty=o_p\!\left(
	\left(\log2p\right)^{-1/2}
	\right).$
Combining this coupling with
\eqref{eq:pseudo-augmented-linearization} yields
\begin{equation}\label{eq:pseudo-final-coupling}
	\|Z-\mu_M-T_M\|_\infty
	=
	o_p\!\left(
	\left(\log2p\right)^{-1/2}
	\right).
\end{equation}
The possibly nonzero inactive mean $\nu_{-M}$ enters only through
the deterministic centering $\mu_M$; the Gaussian coupling
itself is applied to the centered pseudo-score.
Write $G_M\equiv\mu_M+T_M$ for the coupled Gaussian vector.
	
\paragraph*{Step 4: Estimated covariance and threshold.}
Equation~\eqref{eq:pseudo-Hessian-full-rate}, the empirical
gradient condition in Assumption~\ref{A: AsyValid}\ref{SIpivot rho},
and \eqref{eq:restricted-simplified} imply
\begin{equation}\label{eq:pseudo-Q-consistency}
	\|\hat Q_M-Q_M\|_2
	=
	O_p\!\left(
	\sqrt{\frac pn}+m_n\lambda
	\right)
	=o_p(1).
\end{equation}
By Assumption~\ref{A: AsyValid}\ref{SIlocal variance} and the
uniform second-moment bound on $\tau_{i,M}$,
$\|\Omega_M\|_2=O(1)$, while
Assumption~\ref{A: score covariance} implies
$\|\hat\Omega_M\|_2=O_p(1)$. Together with
\eqref{eq:Q-population-bound} and
\eqref{eq:pseudo-Q-consistency},
\begin{align}
	\|\hat\Sigma_M-\Sigma_M\|_2
	&\leq
	\|\hat Q_M-Q_M\|_2
	\|\hat\Omega_M\|_2
	\|\hat Q_M\|_2+
	\|Q_M\|_2
	\|\hat\Omega_M-\Omega_M\|_2
	\|\hat Q_M\|_2
	\notag\\
	&\quad+
	\|Q_M\|_2
	\|\Omega_M\|_2
	\|\hat Q_M-Q_M\|_2
	\notag\\
	&=
	O_p\!\left(
	\sqrt{\frac pn}
	+m_n\lambda
	+r_{I,n}
	\right)
	=o_p(1).
	\label{eq:pseudo-Sigma-consistency}
\end{align}
The variance of the target coordinate is bounded away from zero by
Assumption~\ref{A: AsyValid}\ref{SIpivot rho}:
\begin{equation}\label{eq:pseudo-target-variance}
	\sigma_M^2=\bm e_1'\Sigma_M\bm e_1
	\geq c_\sigma
\end{equation}
for some constant $c_\sigma>0$ and all sufficiently large $n$.
This condition also covers the possibility of covariance between the
empirical AME and the parameter estimation component; for a parameter
functional it follows from the original score and curvature bounds.
Consequently,
\begin{equation*}
	\left|
	\bm e_1'\hat\Sigma_M\bm e_1
	-\sigma_M^2
	\right|
	\leq
	\|\hat\Sigma_M-\Sigma_M\|_2
	=
	O_p\!\left(
	\sqrt{\frac pn}
	+m_n\lambda
	+r_{I,n}
	\right).
\end{equation*}
Moreover, since $
c_M
=
\Sigma_M\bm e_1
\left(
\bm e_1'\Sigma_M\bm e_1
\right)^{-1},$
and the denominator is bounded away from zero by
\eqref{eq:pseudo-target-variance}, the map $
\Sigma
\mapsto
\frac{\Sigma\bm e_1}{\bm e_1'\Sigma\bm e_1}$
is locally Lipschitz at $\Sigma_M$. Hence,
\begin{equation}\label{eq:pseudo-c-consistency}
	\|\hat c_M-c_M\|_2
	=
	O_p\!\left(
	\sqrt{\frac pn}
	+m_n\lambda
	+r_{I,n}
	\right)
	=o_p(1).
\end{equation}
Let $b_M$ denote the population version of $\hat b$ in
\eqref{def: Ab}, obtained by replacing $\hat H_M(\hat\theta_M)$
and $\hat H_{-M,M}(\hat\theta_M)$ with $H_M$ and $H_{-M,M}$.
Since $
\|(\bm0_d',\sm')'\|_2=\sqrt{|M|-d}$ and $n^{1/2}\lambda=C\sqrt{\log2p},$
\eqref{eq:pseudo-H-inverse-difference} and
\eqref{eq:pseudo-H-cross-difference} give
\begin{align}
	\|\hat b-b_M\|_\infty
	&=
	O_p\!\left(
	r_{H,n}\sqrt{m_n\log2p}
	\right)=
	O_p\!\left(
	\sqrt{\frac{pm_n\log2p}{n}}
	+
	\frac{m_n^{3/2}\log2p}{n^{1/2}}
	\right)
	=o_p(1).
	\label{eq:pseudo-threshold-consistency}
\end{align}
Finally, the rate conditions of Proposition~\ref{prop: SI validity},
the covariance rate in
Assumption~\ref{A: score covariance}, and
\eqref{eq:pseudo-final-coupling} imply
\begin{equation}\label{eq:pseudo-joint-approximation-rate}
	\sqrt{\log2p}\,
	\max\{
		\|Z-G_M\|_\infty,
		\|\hat\Sigma_M-\Sigma_M\|_2,
		\|\hat b-b_M\|_\infty,
		\|\hat c_M-c_M\|_2\}
	=o_p(1).
\end{equation}
	
\paragraph*{Step 5: Transfer of the Gaussian pivot.}
Write 
$E_{M,\sm}^G\equiv\left\{AG_M\leq b_M\right\}$.
For the truncation limits in \eqref{eq:Vminus}--\eqref{eq:Vzero},
write $\mathcal V^\pm_{b,c}(r)$ and $\mathcal V^0_{b,c}(r)$ when
the threshold and projection vector need to be shown explicitly.
Recall that $
G_M
=
\mu_M+T_M.$
By Lemma~\ref{lem: Poly},
\begin{align}
	\mathcal P_{M,n}^G
	\equiv
	F\Bigl(
	&\bm e_1'G_M;
	\bm e_1'\mu_M,
	\sigma_M^2, \mathcal{V}^-_{b_M,c_M}
	\left(r(G_M,\Sigma_M,\bm e_1)\right),
	\mathcal V^+_{b_M,c_M}
	\left(r(G_M,\Sigma_M,\bm e_1)\right)
	\Bigr)
	\label{eq:exact-Gaussian-pivot}
\end{align}
satisfies $
	\mathcal P_{M,n}^G\,\vert\,E_{M,\sm}^G
	\sim U(0,1).$
\paragraph*{Substep 5.1: Selection-event approximation.}
We first establish an anti-concentration
bound to show that the feasible and Gaussian
selection events agree asymptotically. Let $a_{j,M}'$ be the row $j$ of $A$.
Since the first column of $A$ is zero and the last
column of $Q_M$ is $\bm e_1$, the empirical-averaging component
does not enter the selection constraints. Writing $Q_M^s$ for
the first $d+p$ columns of $Q_M$, we have
\begin{equation*}
	a_{j,M}'\Sigma_Ma_{j,M}
	=a_{j,M}'Q_M^s I_M{Q_M^s}'a_{j,M}
	\geq c_I\|a_{j,M}'Q_M^s\|_2^2.
\end{equation*}
For an active constraint,
$a_{j,M}'Q_M$ is, up to sign, a row of
$PH_M^{-1}$. Since
$1\leq\|e_j'PH_M^{-1}\|_2\,\|H_M\|_2$
and
$\|H_M\|_2\leq\|H(\theta_M^*)\|_2\leq C_H$,
$\|e_j'PH_M^{-1}\|_2\geq C_H^{-1}.$
For an inactive constraint,
$a_{j,M}'Q_M$ contains a unit coordinate in its inactive
block, so its Euclidean norm is at least one. Thus, there is a
constant $c_A>0$, independent of $M$ and $n$, such that
$a_{j,M}'\Sigma_Ma_{j,M}\geq c_A.$
Assumption~\ref{A: AsyValid}
\ref{SIlocal variance} and
\eqref{eq:Q-population-bound} give
$a_{j,M}'\Sigma_Ma_{j,M}
\leq C_I\|Q_M\|_2^2\|a_{j,M}\|_2^2
\leq C_A$,
since each row of $A$ has unit Euclidean norm.
Hence, for constants $0<c_A<C_A<\infty$,
\begin{equation}\label{eq:Gaussian-constraint-variances}
	c_A
	\leq
	a_{j,M}'\Sigma_Ma_{j,M}
	\leq
	C_A,
	\qquad
	1\leq j\leq 2(d+p)-|M|-d.
\end{equation}
Set
\begin{equation*}
	U_{M,n}^G
	\equiv
	\max_{1\leq j\leq 2(d+p)-|M|-d}
	\left(a_{j,M}'G_M-b_{M,j}\right),
	\qquad
	\hat U_{M,n}
	\equiv
	\max_{1\leq j\leq 2(d+p)-|M|-d}
	\left(a_{j,M}'Z-\hat b_j\right).
\end{equation*}
Then $
E_{M,\sm}^G
=
\{U_{M,n}^G\leq0\},$
and, up to the boundary of the polyhedral representation,
$
E_{M,\sm}
=
\{\hat U_{M,n}\leq0\}.$
To control that boundary, define the centered Gaussian vector
\begin{equation*}
	Y_M
	\equiv
	\left(
	a_{1,M}'T_M,
	\ldots,
	a_{2(d+p)-|M|-d,M}'T_M
	\right)',
\end{equation*}
where $T_M$ is as in \eqref{eq:pseudo-Yurinskii-rate}. 
By \eqref{eq:Gaussian-constraint-variances}, every coordinate of
$Y_M$ has variance bounded below by $c_A$. Nazarov's Gaussian
anti-concentration inequality, in the hyperrectangle form of
\citet[Lemma A.1]{Chernozhukov-Chetverikov-Kato(2017)}, therefore
implies that, uniformly over $y\in\mathbb R^{2(d+p)-|M|-d}$ and
$\varepsilon>0$,
\begin{equation*}
P\left[Y_M\leq y+\varepsilon\bm 1\right]
-P\!\left[Y_M\leq y\right]
\leq
C\varepsilon\sqrt{\log2p},
\end{equation*}
because $2(d+p)-|M|-d=O(p)$. 
Taking $y=b_M-A\mu_M-\varepsilon\bm 1$
and applying the preceding inequality with
$2\varepsilon$ in place of $\varepsilon$ yields
\begin{equation}\label{eq:Gaussian-boundary-anti-concentration}
	P\left[|U_{M,n}^G|\leq\varepsilon\right]
	\leq C\varepsilon\sqrt{\log 2p},
\end{equation}
where the factor of two is absorbed into $C$. Remark that every row of $A$ has exactly one nonzero entry, equal to $1$ or
$-1$, hence $
	|\hat U_{M,n}-U_{M,n}^G|
	\leq
	\|Z-G_M\|_\infty
	+
	\|\hat b-b_M\|_\infty.$
From \eqref{eq:pseudo-joint-approximation-rate}, there exists a
deterministic sequence $\gamma_n\downarrow0$ such that $\gamma_n\sqrt{\log2p}
\to0,$
\begin{equation}\label{eq:shrinking-boundary-tolerance}
	\max\{
		\|Z-G_M\|_\infty,
		\|\hat\Sigma_M-\Sigma_M\|_2,
		\|\hat b-b_M\|_\infty,
		\|\hat c_M-c_M\|_2\}
	=
	o_p(\gamma_n).
\end{equation}
Therefore, letting $\mathbin{\triangle}$ denote the symmetric difference of sets, we have
\begin{align}
	P\!\left[E_{M,\sm}\mathbin{\triangle}
	E_{M,\sm}^G
	\right]
	&\leq
	P\!\left[
	|\hat U_{M,n}-U_{M,n}^G|>\gamma_n
	\right]
	+
	P\!\left[
	|U_{M,n}^G|\leq\gamma_n
	\right]\to0.
	\label{eq: Gaussian selection event approximation}
\end{align}
Thus, the feasible and Gaussian selection events agree
asymptotically. Since $
P[E_{M,\sm}]$
is bounded away from zero by
Assumption~\ref{A: SI pivot}\ref{SIpivot selection},
\eqref{eq: Gaussian selection event approximation} also gives
\begin{equation}\label{eq:selection-probability-ratio}
\liminf_{n\to\infty}P[E_{M,\sm}^G]>0,
\qquad
\frac{P[E_{M,\sm}^G]}{P[E_{M,\sm}]}
\to1.
\end{equation}

\paragraph*{Substep 5.2: Truncation nondegeneracy.}
We next show that, conditional on Gaussian selection,
the truncation interval has sufficient probability
mass and the constraints with zero population slopes
have sufficient slack. These properties ensure that
the truncated Gaussian pivot is stable under the
approximation errors.
The Gaussian mass of the truncation interval is
\begin{equation}\label{eq: truncation mass definition}
    \mathcal D_M(r)\equiv
    \Phi\!\left(\frac{\mathcal V^+_{b_M,c_M}(r)-\bm e_1'\mu_M}
                    {\sigma_M}\right)
    -\Phi\!\left(\frac{\mathcal V^-_{b_M,c_M}(r)-\bm e_1'\mu_M}
                     {\sigma_M}\right).
\end{equation}
Write $r_M^G= r(G_M,\Sigma_M,\bm e_1)$ and take
$\mathcal V^0_{b_M,c_M}(r)=+\infty$ when no row $j$ has $(Ac_M)_j=0$.  
 Conditional on $r_M^G$, Lemma~\ref{lem: Poly} implies
$P\left[E_{M,\sm}^G
	\,\vert\,
	r_M^G
	\right]
	=
	1\!\left(
	\mathcal V^0_{b_M,c_M}
	(r_M^G)\geq0
	\right)
	\mathcal D_M(r_M^G).$
Consequently,
\begin{align*}
	P\!\left[
	\mathcal D_M(r_M^G)<\gamma_n,
	\,
	E_{M,\sm}^G
	\right]
	&=
	\E\!\left[
	1\!\left(
	\mathcal D_M(r_M^G)<\gamma_n
	\right)
	1\!\left(
	\mathcal V^0_{b_M,c_M}
	(r_M^G)\geq0
	\right)
	\mathcal D_M(r_M^G)
	\right]\leq\gamma_n.
\end{align*}
Hence,
\begin{equation}\label{eq:truncation-mass-lower-bound}
	P\!\left[
	\mathcal D_M(r_M^G)<\gamma_n
	\,\middle|\,
	E_{M,\sm}^G
	\right]
	\leq
	\frac{\gamma_n}
	{P[E_{M,\sm}^G]}
	\to0.
\end{equation}
It remains to control the rows with zero population slope. Let $
	\mathcal B
	\equiv
	\{
	j:(A c_M)_j=0
	\}.$
If $\mathcal B\neq\varnothing$, then, because
$a_{j,M}'c_M=0$ for $j\in\mathcal B$,
\begin{align*}
	-\mathcal V^0_{b_M,c_M}(r_M^G)
	&=
	\max_{j\in\mathcal B}
	\{
	a_{j,M}'r_M^G-b_{M,j}
	\}=
	\max_{j\in\mathcal B}
	\{
	a_{j,M}'G_M-b_{M,j}
	\}.
\end{align*}
The subvector
$\{a_{j,M}'T_M:j\in\mathcal B\}$ again has coordinate
variances bounded below by $c_A$. Applying Nazarov's
inequality again gives
\begin{equation*}
	P\!\left[
	0\leq
	\mathcal V^0_{b_M,c_M}(r_M^G)
	<\gamma_n,
	\,
	E_{M,\sm}^G
	\right]
	\leq
	C\gamma_n\sqrt{\log 2p}.
\end{equation*}
Dividing by $P[E_{M,\sm}^G]$ and using
\eqref{eq:selection-probability-ratio},
\begin{equation}\label{eq:zero-slope-slack}
	P\!\left[
	0\leq
	\mathcal V^0_{b_M,c_M}(r_M^G)
	<\gamma_n
	\,\middle|\,
	E_{M,\sm}^G
	\right]
	\to0.
\end{equation}
When $\mathcal B=\varnothing$, the same conclusion is automatic
under the convention
$\mathcal V^0_{b_M,c_M}=+\infty$.
Combining
\eqref{eq:truncation-mass-lower-bound} and
\eqref{eq:zero-slope-slack} yields
\begin{equation}\label{eq:derived-truncation-stability}
	P\!\left[
	\mathcal D_M(r_M^G)\geq\gamma_n,
	\ 
	\mathcal V^0_{b_M,c_M}(r_M^G)
	\geq\gamma_n
	\,\middle|\,
	E_{M,\sm}^G
	\right]
	\to1.
\end{equation}
\paragraph*{Substep 5.3: Pivot transfer.}
Let $\mathcal P_{M,n}$ denote the feasible pivot
evaluated at $\rho_M(\theta_M^*)$.
By \eqref{eq:shrinking-boundary-tolerance},
Assumption~\ref{A: SI pivot}\ref{SIpivot poly}
implies that, for every $j$ with $(Ac_M)_j\neq0$,
$\operatorname{sign}((A\hat c_M)_j)
=\operatorname{sign}((Ac_M)_j)$ and
$|(A\hat c_M)_j|\geq\kappa/2$ wpa 1.
For these rows,
\begin{equation}\label{eq:truncation-bound-identity}
	\frac{b_{M,j}-(Ar_M^G)_j}{(Ac_M)_j}
	=
	\bm e_1'G_M+
	\frac{b_{M,j}-(AG_M)_j}{(Ac_M)_j},
\end{equation}
and the analogous identity holds for
$\hat b$, $\hat r$, $\hat c_M$, and $Z$.
Since $\bm e_1'(G_M-\mu_M)=O_p(1)$,
the Gaussian density is bounded, and
$\sup_{t\in\mathbb R}|t|\phi(t)<\infty$,
the Gaussian CDF evaluated at these truncation
bounds is uniformly Lipschitz in the target
coordinate, constraint slack, and nonzero slope,
with an $O_p(1)$ Lipschitz bound.
Thus, the corresponding Gaussian CDF values
differ by $o_p(1)$.

For rows with $(Ac_M)_j=0$, on the event in
\eqref{eq:derived-truncation-stability},
$b_{M,j}-(AG_M)_j\geq\gamma_n$.
By \eqref{eq:shrinking-boundary-tolerance},
$\hat b_j-(AZ)_j\geq\gamma_n-o_p(\gamma_n)$
and $(A\hat c_M)_j=o_p(\gamma_n)$.
Consequently, whenever $(A\hat c_M)_j\neq0$,
\begin{equation*}
\left|
\frac{\hat b_j-(A\hat r)_j}{(A\hat c_M)_j}
-\bm e_1'Z
\right|
=
\left|
\frac{\hat b_j-(AZ)_j}{(A\hat c_M)_j}
\right|
\longrightarrow_p\infty.
\end{equation*}
These rows therefore contribute only
asymptotically infinite truncation bounds. Together with
$\|\hat\Sigma_M-\Sigma_M\|_2=o_p(\gamma_n)$,
these results imply that the numerator and
normalizing probabilities of the feasible
and Gaussian truncated CDFs differ by
$o_p(1)$, conditionally on $E_{M,\sm}^G$. For any fixed $\varepsilon>0$, the argument
underlying \eqref{eq:truncation-mass-lower-bound}
gives
\begin{equation*}
P\left[
\mathcal D_M(r_M^G)<\varepsilon
\,\middle|\,E_{M,\sm}^G
\right]
\leq
\frac{\varepsilon}{P[E_{M,\sm}^G]}.
\end{equation*}
On $\mathcal D_M(r_M^G)\geq\varepsilon$,
the feasible normalizing probability is
at least $\varepsilon/2$ wpa 1.
The $o_p(1)$ differences in the numerator
and denominator therefore imply
$|\mathcal P_{M,n}-\mathcal P_{M,n}^G|
=o_p(1)$ on this event.
Letting $\varepsilon\downarrow0$ and using
\eqref{eq:selection-probability-ratio} and
\eqref{eq:derived-truncation-stability} yields
\begin{equation}\label{eq:feasible-Gaussian-pivot-close-G}
	|\mathcal P_{M,n}-\mathcal P_{M,n}^G|
	=o_p(1)
	\quad\text{conditionally on }E_{M,\sm}^G.
\end{equation}
By
\eqref{eq: Gaussian selection event approximation},
\eqref{eq:selection-probability-ratio}, and the nonvanishing
selection probability, the same conclusion holds conditional on the
feasible selection event:
\begin{equation}\label{eq:feasible-Gaussian-pivot-close}
	P\!\left[
	|\mathcal P_{M,n}-\mathcal P_{M,n}^G|>\epsilon
	\,\middle|\,
	E_{M,\sm}
	\right]
	\to0
\end{equation}
for every $\epsilon>0$. By \eqref{eq: Gaussian selection event approximation}
and the nonvanishing selection probability,
for any fixed $t\in(0,1)$,
\begin{equation}
	\left|
	P[\mathcal P_{M,n}^G\leq t\,\vert\, E_{M,\sm}]
	-
	P[\mathcal P_{M,n}^G\leq t\,\vert\, E_{M,\sm}^G]
	\right|\leq
	2\frac{
		P[E_{M,\sm}\mathbin{\triangle}E_{M,\sm}^G]
	}{
		P[E_{M,\sm}]
	}
	=o(1).
	\label{eq:conditional-event-transfer}
\end{equation}
Therefore, for fixed $t\in(0,1)$ and $\epsilon>0$,
\begin{align*}
	\left|
	P[
	\mathcal P_{M,n}\leq t
	\,\vert\, E_{M,\sm}
	]
	-t
	\right|
	&\leq
	P\!\left[
	|\mathcal P_{M,n}-\mathcal P_{M,n}^G|>\epsilon
	\,\vert\, E_{M,\sm}
	\right]
	+2\epsilon
	+
	2\,
	\frac{
		P[
		E_{M,\sm}
		\mathbin{\triangle}
		E_{M,\sm}^G
		]
	}{
		P[E_{M,\sm}]
	}.
\end{align*}
The right-hand side converges to $2\epsilon$ by
\eqref{eq:exact-Gaussian-pivot},
\eqref{eq: Gaussian selection event approximation}, and
\eqref{eq:feasible-Gaussian-pivot-close}.
Letting $\epsilon\downarrow0$ gives $
	P\!\left[
	\mathcal P_{M,n}\leq t
	\,\vert\,E_{M,\sm}
	\right]
	\to t.$
Thus, the feasible pivot is asymptotically uniform conditional on
$\{\hat M=M,\hsm=\sm\}$. Inverting the equal-tailed pivot yields
\begin{equation*}
	P\!\left[
	\rho_M(\theta_M^*)
	\in
	\mathrm{CI}_{\hat M}
	\,\middle|\,
	\hat M=M,\hsm=\sm
	\right]
	\to
	1-\zeta.
\end{equation*}
\subsection{Proof of Corollary~\ref{cor: SI model validity}}
\label{subsec: Proof SI model validity}
Let $q_M=|M|-d$ and
$E_{M,s}\equiv\{\hat M=M,\hsm=s\}$. Define
$
\Delta_{n,s}
\equiv
P\left[
\rho_M(\theta_M^*)\in\mathrm{CI}_{\hat M}
\,\middle|\,E_{M,s}
\right]-(1-\zeta)$
for sign configurations with $P[E_{M,s}]>0$,
and set $\Delta_{n,s}=0$ otherwise.
Also, let
$p_{n,s}\equiv P[E_{M,s}\mid\hat M=M]$.
By the law of total probability,
\begin{equation*}
P\left[\rho_M(\theta_M^*)\in\mathrm{CI}_{\hat M}
\,\middle|\,\hat M=M
\right]-(1-\zeta)=
\sum_{s\in\{-1,1\}^{q_M}}p_{n,s}\Delta_{n,s}.
\end{equation*}
Let $\bar q\equiv\sup_n q_M<\infty$. Changing the sign vector only
changes the signs of some rows of $A$, so $|(Ac_M)_j|$ and hence
Assumption~\ref{A: SI pivot}\ref{SIpivot poly} do not depend on the
sign vector.
For any $\varepsilon>0$, since
$\sum_s p_{n,s}=1$ and $|\Delta_{n,s}|\leq1$,
we have
\begin{align*}
	\left|\sum_s p_{n,s}\Delta_{n,s}\right|
	\leq
	\sum_{s:p_{n,s}>\varepsilon}
	p_{n,s}|\Delta_{n,s}|
	+
	\sum_{s:p_{n,s}\leq\varepsilon}
	p_{n,s}|\Delta_{n,s}|\leq
	\sup_{s:p_{n,s}>\varepsilon}|\Delta_{n,s}|
	+2^{\bar q}\varepsilon,
\end{align*}
where the supremum over an empty set is taken to be zero. We show that the supremum converges to zero.
Otherwise, there exist a subsequence and corresponding sign configurations $s_n$ such that $p_{n,s_n}>\varepsilon$ and $|\Delta_{n,s_n}|$ remains bounded away from zero. Since $q_M\leq\bar q$, only finitely many sign configurations are possible, so at least one sign vector, say $s$, must occur infinitely often along this subsequence. Restricting to those terms yields a subsequence with a fixed sign vector $s$. By the assumption on $P[\hat M=M]$, $
P[E_{M,s}]=p_{n,s}P[\hat M=M]$
is bounded away from zero along this sequence.
The proof of Proposition~\ref{prop: SI validity},
applied along the same sequence, then gives
$\Delta_{n,s}\to0$, which is a contradiction. Consequently,
\begin{equation*}
	\limsup_{n\to\infty}
	\left|
	P\left[
	\rho_M(\theta_M^*)\in\mathrm{CI}_{\hat M}
	\,\middle|\,\hat M=M
	\right]-(1-\zeta)
	\right|
	\leq2^{\bar q}\varepsilon.
\end{equation*}
Since $\varepsilon>0$ is arbitrary, the result follows.
\section{Supplementary result}\label{sec: supplementary result}

We use the following result from
\cite{Cattaneo2024}.

\begin{theorem}[Yurinskii coupling for independent vectors, Corollary 2.3 of \cite{Cattaneo2024}]
	\label{thm: Yur}
	Let $U_1,\ldots,U_n$ be independent, mean-zero
	random vectors in $\mathbb R^k$, and let
	$V_i=\operatorname{Var}(U_i)$ and
	$\Sigma=\sum_{i=1}^n V_i$.
	Define
	\begin{align*}
		\beta_{\infty,2}
		\equiv\sum_{i=1}^n\E\left[
		\|U_i\|_2^2\|U_i\|_\infty
		+
		\|V_i^{1/2}Z_i\|_2^2
		\|V_i^{1/2}Z_i\|_\infty
		\right],
	\end{align*}
	where $Z_1,\ldots,Z_n$ are independent
	standard Gaussian vectors in $\mathbb R^k$.
	For every $\eta>0$, there exists a Gaussian
	vector $T\sim N(0,\Sigma)$ such that
	\begin{equation*}
		P\left[
		\left\Vert\sum_{i=1}^nU_i-T\right\Vert_\infty
		>\eta
		\right]
		\leq
		24\left(
		\frac{\beta_{\infty,2}\phi_\infty(k)^2}
		{\eta^3}
		\right)^{1/3},
	\end{equation*}
	where $\phi_\infty(k)=\sqrt{2\log2k}$,
	provided $\beta_{\infty,2}<\infty$.
\end{theorem}
\end{document}